\documentclass{toc}
\usepackage[nottoc]{tocbibind} 

\tocdetails{%
volume=10, number=11, year=2014, firstpage=257,
received={February 16, 2013},
revised={August 11, 2014},
published={October 2, 2014},
doi={10.4086/toc.2014.v010a011}
}
\runningtitle{Efficient Rounding for the Noncommutative Grothendieck Inequality}

\tocpdftitle{Efficient Rounding for the Noncommutative Grothendieck Inequality}

\runningauthor{Assaf Naor, Oded Regev, and Thomas Vidick}
\copyrightauthor{Assaf Naor, Oded Regev, and Thomas Vidick}
\tocpdfauthor{Assaf Naor, Oded Regev, Thomas Vidick}

\newcommand{\eqdef}{\stackrel{\mathrm{def}}{=}}
\newenvironment{step}
  {
    \begin{enumerate}

  }
  {\end{enumerate}}

\newenvironment{protocol*}[1]
  {
    \begin{center}
      \hrulefill\\
      \textbf{#1}
  }
  {
    \vspace{-1\baselineskip}
    \hrulefill
    \end{center}
  }

\renewcommand{\le}{\leqslant}
\renewcommand{\ge}{\geqslant}
\renewcommand{\leq}{\leqslant}
\renewcommand{\geq}{\geqslant}

\newcommand{\e}{\varepsilon}
\renewcommand{\d}{\delta}
\newcommand{\Bil}{\operatorname{Bil}_\R}

\newcommand{\poly}{{\mathrm{poly}}}
\newcommand{\Tr}{\mbox{\rm Tr}}

\newcommand{\opt}{\mathrm{Opt}}
\newcommand{\alg}{\mathrm{Alg}}

\newcommand{\nc}[1]{\|#1\|_{\operatorname{nc}}}

\newcommand{\SDP}{\mathrm{SDP}}

\newcommand{\eps}{\varepsilon}
\renewcommand{\O}{\mathcal{O}}
\newcommand{\U}{\mathcal{U}}

\newcommand{\E}{\mathbb{E}}

\newcommand{\Xz}{X_{{z}}}
\newcommand{\Yz}{Y_{{z}}}
\newcommand{\Xe}{X_{{\eps}}}
\newcommand{\Ye}{Y_{{\eps}}}

\begin{document}

\begin{frontmatter}[classification=text]
\title{Efficient Rounding for the\\ Noncommutative Grothendieck Inequality\titlefootnote{An extended abstract of this paper appeared in the 
Proceedings of the 45th Annual ACM Symposium on Theory of Computing, 
2013~\cite{NRV13stoc}.}}
\author[naor]{Assaf Naor\thanks{Supported by NSF grant CCF-0832795, BSF grant
 2010021, the Packard Foundation and the Simons Foundation. Part of this work was completed while A. N. was
visiting Universit\'e de Paris Est Marne-la-Vall\'ee.}}
\author[regev]{Oded Regev\thanks{Supported by a European Research Council (ERC)
    Starting Grant. Part of the work done while the author was with
    the CNRS, DI, ENS, Paris.}}
\author[vidick]{Thomas Vidick\thanks{Partially supported by the National Science Foundation under Grant No. 0844626 and by the Ministry of Education, Singapore under the Tier 3 grant MOE2012-T3-1-009.}}
\begin{abstract}
The classical Grothendieck inequality has applications to the design of approximation algorithms for \cclass{NP}-hard optimization problems. We show that an algorithmic  interpretation may also be given for a \emph{noncommutative} generalization of the Grothendieck inequality  due to Pisier and Haagerup. Our main result, an efficient rounding procedure for this inequality, leads to a 
polynomial-time 
constant-factor %
approximation algorithm for an optimization problem which  generalizes the Cut Norm problem of Frieze and Kannan, and is shown here to have additional applications to robust 
principal 
component analysis and the orthogonal Procrustes problem.
\end{abstract}

\tocacm{G.1.6}
\tocams{68W25}

\tockeywords{approximation algorithms, Grothendieck inequality, semidefinite programming, principal component analysis}
\end{frontmatter}

\newpage
\tableofcontents
\newpage
\section{Introduction}

 In what follows, the standard scalar product on $\C^n$ is denoted $\langle\cdot,\cdot\rangle$, \ie, $\langle x,y\rangle=\sum_{i=1}^n x_i\overline{y_i}$ for all $x,y\in \C^n$. We always think of $\R^n$ as canonically embedded in $\C^n$; in particular the restriction of $\langle\cdot,\cdot\rangle$ to $\R^n$ is the standard scalar product on $\R^n$. Given a set $S$, the space $M_n(S)$ stands for all the matrices $M=(M_{ij})_{i,j=1}^n$ with $M_{ij}\in S$ for all $i,j\in \{1,\ldots,n\}$. Thus,  $M_n(M_n(\R))$ is naturally identified with the $n^4$-dimensional space of all $4$-tensors $M=(M_{ijkl})_{i,j,k,l=1}^n$ with $M_{ijkl}\in \R$ for all $i,j,k,l\in \{1,\ldots,n\}$.  The set of all $n\times n$ orthogonal matrices is denoted $\O_n\subseteq M_n(\R)$, and the  set of all $n\times n$ unitary matrices is denoted $\mathcal{U}_n\subseteq M_n(\C)$.

 Given $M=(M_{ijkl})\in M_n(M_n(\R))$ denote
\[
 \mathrm{Opt}_\R(M)\eqdef \sup_{U,V\in \O_n}\sum_{i,j,k,l=1}^n M_{ijkl}U_{ij}V_{kl}
\]
and, similarly, for $M=(M_{ijkl})\in M_n(M_n(\C))$ denote
\[
 \opt_\C(M)\eqdef \sup_{U,V\in \U_n}\Big|\sum_{i,j,k,l=1}^n M_{ijkl}U_{ij}\overline{V_{kl}}\Big|\,.
\]

 \begin{theorem}\label{thm:alg}
 There exists a 
polynomial-time 
algorithm that takes as input  $M\in M_n(M_n(\R))$  and 
outputs  %
$U,V\in \O_n$  such that 
\[
\mathrm{Opt}_\R(M)\le O(1)
\sum_{i,j,k,l=1}^n M_{ijkl}U_{ij}V_{kl}\,.
\]
Analgously,
there exists a 
polynomial-time 
algorithm that takes as input  $M\in M_n(M_n(\C))$ and 
outputs  %
$U,V\in \U_n$ such that $$\opt_\C(M)\le O(1)\Big|\sum_{i,j,k,l=1}^n M_{ijkl}U_{ij}\overline{V_{kl}}\Big|\,.$$
 \end{theorem}

We will explain the ideas that go into the proof of
\expref{Theorem}{thm:alg} later, and it suffices to say at this
juncture that our algorithm is based on a rounding procedure for
semidefinite programs that is markedly different from rounding
algorithms that have been previously used in the optimization
literature, and as such it indicates the availability of techniques
that have thus far remained untapped for the purpose of algorithm
design. A non-constructive version of 
\expref{Theorem}{thm:alg}
was 
given
in earlier works, as explained below; our contribution here is to design an efficient rounding algorithm and to establish various applications of it.
Before
explaining the proof of \expref{Theorem}{thm:alg} we list
some of its applications as an indication of its usefulness.

 \begin{remark}\label{rem:constants}
 The implied constants in the $O(1)$ terms of \expref{Theorem}{thm:alg} can be taken to be any number greater than $2\sqrt{2}$ in the real case, and any number greater than $2$ in the complex case. There is no reason to believe that the factor $2\sqrt{2}$ in the real case is optimal, but the factor $2$ in the complex case is sharp in a certain natural sense that will become clear later. The main content of \expref{Theorem}{thm:alg} is the availability of a 
constant-factor 
algorithm rather than the value of the constant itself. In particular, the novelty of the applications to combinatorial optimization that are described below is the mere existence of a constant-factor approximation algorithm.
 \end{remark}

 \subsection{Applications of \expref{Theorem}{thm:alg}}\label{sec:apps intro}

We now describe some examples demonstrating the usefulness of \expref{Theorem}{thm:alg}. The first example does not lead to a new result, and is meant to put \expref{Theorem}{thm:alg} in context. All the other examples lead to new algorithmic results.
Many of the applications below follow from a more versatile reformulation of \expref{Theorem}{thm:alg} that is presented in \expref{Section}{sec:applications} (see \expref{Proposition}{prop:usefulform}).

 \subsubsection{The Grothendieck problem}\label{sec:gro alg intro}  The Grothendieck optimization problem takes as input a matrix $A\in M_n(\R)$ and aims to efficiently compute (or estimate) the quantity
 \begin{equation}\label{eq:grothendieck problem formulation}
 \max_{\e,\d\in \{-1,1\}^n}\sum_{i,j=1}^n A_{ij}\e_i\d_j\,.
 \end{equation}
 This problem falls into the framework of \expref{Theorem}{thm:alg} by considering the $4$-tensor $M\in M_n(M_n(\R))$ given by $M_{iijj}=A_{ij}$ and $M_{ijkl}=0$ if either $i\neq j$ or $k\neq l$. Indeed,
 $$
 \opt_\R(M)=\max_{U,V\in \O_n}\sum_{i,j=1}^n A_{ij} U_{ii}V_{jj}= 
\max_{\e,\d\in \{-1,1\}^n}\sum_{i,j=1}^n A_{ij}\e_i\d_j\,.
 $$
Here we used the diagonal structure of the tensor $M$ to argue that one may without loss of generality restrict the variables $U,V\in\O_n$ to be diagonal $\pm 1$ matrices.
While diagonal matrices always commute, in general $U,V\in O_n(\R)$ will not, and it is in this sense that we may think of $\opt_\R(\cdot)$ as a non-commutative generalization of~\eqref{eq:grothendieck problem formulation}. 

A 
polynomial-time 
constant-factor %
approximation algorithm for the Grothendieck problem was designed in~\cite{AN04}, where it was also shown that it is \cclass{NP}-hard to approximate this problem within a factor less that $1+\e_0$ for some $\e_0\in (0,1)$. A simple transformation~\cite{AN04} relates the Grothendieck problem to the Frieze-Kannan Cut Norm problem~\cite{FK99} (this transformation can be made to have no loss in the approximation guarantee~\cite[Sec.~2.1]{KhotN11survey}), and as such the constant-factor approximation algorithm for the Grothendieck problem has found a variety of applications in combinatorial optimization; see the survey~\cite{KhotN11survey} for much more on this topic. In another direction, based on important work of Tsirelson~\cite{Tsirelson:85b}, the Grothendieck problem has found applications to quantum information theory~\cite{CHTW04}.
Since the problem of computing $\opt_\R(\cdot)$ contains the Grothendieck problem as a special case, \expref{Theorem}{thm:alg} encompasses all of these applications, albeit with the approximation factor being a larger constant.

 \subsubsection{Robust PCA}\label{sec:PCA intro}

The input to the classical principal component analysis (PCA) problem is $K,n\in \N$  a set of points $a_1,\ldots,a_N \in \R^n$. The goal is to  find a $K$-dimensional subspace maximizing the sum of the squared $\ell_2$ norms of the projections of the $a_i$ on
the subspace. Equivalently, the problem is to find the maximizing vectors in
\begin{equation}\label{eq:classical PCA}
 \max_{\substack{y_1,\ldots,y_K \in \R^n\\ \langle y_i,y_j\rangle = \delta_{ij}}} \,\sum_{i=1}^N \sum_{j=1}^K \,|\langle a_i, y_j \rangle|^2\,,
\end{equation}
where here, and in what follows, $\d_{ij}$ is the Kronecker delta.
This question has a closed-form solution in terms of the singular values of the $N\times n$ matrix whose $i$-th row contains the coefficients of the point $a_i$.

The fact that the quantity appearing in~\eqref{eq:classical PCA} is the maximum of the sum of the \emph{squared} norms of the projected points makes it somewhat non-robust to outliers,
in the sense that a single long vector can have a large effect on the maximum. Several more robust versions of PCA
were suggested in the literature. One variant, known  as ``R1-PCA,'' is due to Ding, Zhou, He, and Zha~\cite{dingetal},
and aims to maximize the sum of the Euclidean norms of the projected points, namely,
\begin{equation}\label{eq:R1PCA def}
 \max_{\substack{y_1,\ldots,y_K \in \R^n\\ \langle y_i,y_j\rangle = \delta_{ij}}}\, \sum_{i=1}^N \Big(\sum_{j=1}^K \,|\langle a_i, y_j \rangle|^2\Big)^{1/2}\,.
\end{equation}
We are not aware of any prior efficient algorithm for this problem that achieves a guaranteed approximation factor.
Another robust variant of PCA, known as ``L1-PCA,'' was suggested by Kwak~\cite{kwak}, and further
studied by McCoy and Tropp~\cite{MT12} (see Section~2.7 in~\cite{MT12} in particular). Here the goal
is to maximize the sum of the $\ell_1$ norms of the projected points, namely,
\begin{equation}\label{eq:L1PCA def}
 \max_{\substack{y_1,\ldots,y_K \in \R^n\\ \langle y_i,y_j\rangle = \delta_{ij}}}\, \sum_{i=1}^N \sum_{j=1}^K\, |\langle a_i, y_j \rangle|\,.
\end{equation}
In~\cite{MT12} a 
constant-factor 
approximation algorithm for the above problem is obtained for $K=1$ based on~\cite{AN04}, and for general $K$ an approximation algorithm with an approximation guarantee of $O(\log n)$ is obtained based on prior work by So~\cite{So11}.

In \expref{Section}{sec:app_pca} we show that both of the above robust versions of PCA can be cast as special cases of \expref{Theorem}{thm:alg}, thus yielding constant-factor approximation algorithms for both problems and all $K\in \{1,\ldots,n\}$. The key observation is that both~\eqref{eq:R1PCA def} and~\eqref{eq:L1PCA def} can be expressed as bilinear optimization problems over a pair of orthogonal matrices. For instance, in the case of~\eqref{eq:L1PCA def} one of the matrix variables will have the $y_i$ as its columns, while the other will be diagonal, with diagonal entries (in an optimal solution) matching the sign of $\langle a_i,y_j\rangle$.

 \subsubsection{The orthogonal Procrustes problem}\label{sec:procustes intro}

Let $n,d\geq 1$ and $K\geq 2$ be integers. Suppose that $S_1,\ldots,S_K\subseteq \R^d$ are $n$-point subsets of $\R^d$. The goal of the \emph{generalized orthogonal Procrustes problem} is to rotate each of the $S_k$ separately so as to best align them. Formally, write $S_k=\{x^k_1,x_2^k,\ldots,x_n^k\}$. The goal is to find $K$ orthogonal matrices $U_1,\ldots,U_K\in \O_d$ that maximize the quantity
\begin{equation}\label{eq:sum of averages}
\sum_{i=1}^n\Big\|\sum_{k=1}^K U_k x^k_i \Big\|_2^2\,.
\end{equation}

 If one focuses on a single summand appearing in~\eqref{eq:sum of
   averages}, say $\sum_{k=1}^K U_k x^k_1$,  then it is clear that in
 order to maximize its length one would want to rotate each of the
 $x_1^k$ so that they would all point in the same direction, \ie, they
 would all be positive multiples of the same vector. The above problem
 aims to achieve the best possible such alignment (in aggregate) for
 multiple summands of this type. We note that by expanding the squares
 one sees that $U_1,\ldots,U_K\in \O_d$ maximize the quantity
 appearing in~\eqref{eq:sum of averages} if and only if they minimize
 the quantity
\[
\sum_{i=1}^n\sum_{k,l=1}^K \bigl\|U_kx^k_i-U_l x^l_i\bigr\|_2^2\,.
\]

 The term ``generalized'' was used above because the \emph{orthogonal Procrustes problem} refers to the case $K=2$, which has a closed-form solution. The name ``Procustes'' is a (macabre) reference to Greek mythology (see \url{http://en.wikipedia.org/wiki/Procrustes}).

 The generalized orthogonal Procrustes problem has been extensively studied since the 1970s, initially in the psychometric literature (see, \eg,~\cite{CC71,Gow75,Ber77}), and more recent applications of it are to areas such as image and shape analysis, market research, %
 and biometric identification; see the books~\cite{GD04,DM98}, the lecture notes~\cite{SG02}, and~\cite{MP07} for much more information on this topic.

 The generalized orthogonal Procrustes problem is known to be intractable, and it has been investigated algorithmically in, \eg,~\cite{Ber77,Bou82,SB88}. A rigorous analysis of a polynomial-time approximation algorithm for this problem appears in the work of Nemirovski~\cite{Nem07}, where the generalized orthogonal Procrustes problem is treated as an important special case of a more general family of problems called ``quadratic optimization under orthogonality constraints,'' for which he obtains a $O(\sqrt[3]{n+d}+\log K)$ approximation algorithm. This was subsequently improved by So~\cite{So11} to $O(\log(n+d+K))$. In \expref{Section}{sec:app_procrustes} we use \expref{Theorem}{thm:alg} to improve the approximation guarantee for the generalized orthogonal Procrustes problem as defined above to a constant approximation factor.  See also \expref{Section}{sec:app_procrustes} for a more complete discussion of variants of this problem considered in~\cite{Nem07,So11} and how they compare to our work.

 \subsubsection{A  Frieze-Kannan decomposition for \texorpdfstring{$4$}{4}-tensors}

 In~\cite{FK99}  Frieze and Kannan designed an algorithm which decomposes every (appropriately defined) ``dense'' matrix into a sum of a few ``cut matrices'' 
plus an error matrix that has small cut-norm.  We refer to~\cite{FK99} and also Section~2.1.2 in the survey~\cite{KhotN11survey} for a precise formulation of this statement, as well as its extension, due to~\cite{AN04}, to an algorithm that allows sub-constant errors.  In \expref{Section}{sec:regularity} we apply \expref{Theorem}{thm:alg} to prove the following result, which can be viewed as a noncommutative variant of the Frieze-Kannan decomposition. For the purpose of the statement below it is convenient to identify the space $M_n(M_n(\C))$ of all $4$-tensors with $M_n(\C)\otimes M_n(\C)$. Also, for $M\in M_n(\C)\otimes M_n(\C)$ we denote from now on its Frobenius (Hilbert-Schmidt) norm by
 $$
 \|M\|_2\eqdef\sqrt{\sum_{i,j,k,l=1}^n |M_{ijkl}|^2}\,.
 $$

 \begin{theorem}\label{thm:regularity} There exists a universal constant $c\in (0,\infty)$ with the following property. Suppose that $M\in M_n(\C)\otimes M_n(\C)$ and $0<\eps\leq 1/2$, and let
 \begin{equation}\label{eq:def T}
 T \eqdef \left\lceil \frac{c n^2 \|M\|_2^2}{\eps^2  \opt_\C(M)^2}\right\rceil\,.
 \end{equation}
 One can compute in time $\poly(n,1/\eps)$ a decomposition
\begin{equation}\label{eq:M decomposed}
 M \,=\, \sum_{t=1}^T \,\alpha_t (A_t \otimes B_t) + E
 \end{equation}
such that $A_t,B_t \in \U_n$, the coefficients $\alpha_t \in \C$ satisfy $|\alpha_t| \leq \|M\|_2/n$, and $\opt_\C(E) \leq \eps \opt_\C(M)$. Moreover, if $M\in M_n(\R)\otimes M_n(\R)$  then one can replace $\opt_\C(M)$ in~\eqref{eq:def T} by $\opt_\R(M)$, take the coefficients $\alpha_t$ to be real, $A_t,B_t\in \O_n$ and $E$ such that $\opt_\R(E)\leq \eps\opt_\R(M)$.
\end{theorem}

\expref{Theorem}{thm:regularity} contains as a special case its
commutative counterpart, as studied in~\cite{FK99,AN04}. Here we are
given $A\in M_n(\R)$ with $|a_{ij}|\le 1$ for all $i,j\in
\{1,\ldots,n\}$, and we aim for an error  $\e n^2$. Define
$M_{iijj}=a_{ij}$ and $M_{ijkl}=0$ if $i\neq j$ or $k\neq l$. Then
$\|M\|_2\le n$. An application of \expref{Theorem}{thm:regularity} (in
the real case) with $\e$ replaced by $\e n^2/\opt_\R(M)$ yields a
decomposition $A=\sum_{t=1}^T \alpha_t (a_t b_t^*) +E$ with
$a_t,b_t\in [-1,1]^n$ and $E\in M_n(\R)$ satisfying
\[
\sup_{\e,\d\in \{-1,1\}}\sum_{ij=1}^n E_{ij} \e_i\d_j\le \e n^2\,.
\]
Moreover, the number of terms is $T=O(1/\e^2)$.

\expref{Theorem}{thm:regularity} is proved in \expref{Section}{sec:regularity} via an iterative application of \expref{Theorem}{thm:alg}, following the ``energy decrement'' strategy as formulated by Lov{\'a}sz and Szegedy~\cite{LovaszS07} in the context of general weak regularity lemmas. Other than being a structural statement of interest in its own right, we show in \expref{Section}{sec:regularity} that \expref{Theorem}{thm:regularity} can be used to enhance the 
constant-factor 
approximation of \expref{Theorem}{thm:alg} to a PTAS for computing $\opt_\C(M)$ when $\opt_\C(M)=\Omega(n\|M\|_2)$. Specifically, if $\opt_\C(M)\ge \kappa n\|M\|_2$ then one can compute a $(1+\e)$-factor approximation to $\opt_\C(M)$ in time $2^{\poly(1/(\kappa\e))}\poly(n)$. This is reminiscent of the Frieze-Kannan algorithmic framework~\cite{FK99} for dense graph and matrix problems.

\subsubsection{Quantum XOR games}
As we already noted, the Grothendieck problem (recall \expref{Section}{sec:gro alg intro}) also has consequences in quantum information theory~\cite{CHTW04}, and more specifically to bounding the power of entanglement in so-called ``XOR games,'' which are two-player one-round games in which the players each answer with a bit and the referee bases her decision on the XOR of the two bits. As will be explained in detail in \expref{Section}{sec:gro ineq} below, the literature on the Grothendieck problem relies on a classical inequality of Grothendieck~\cite{Gro53}, while our work relies on a more recent yet by now classical noncommutative Grothendieck inequality of Pisier~\cite{Pisier78NCGT} (and its sharp form due to Haagerup~\cite{Haagerup85NCGT}). Even more recently, the Grothendieck inequality has been generalized to another setting, that of \emph{completely bounded} linear maps defined on operator spaces~\cite{PS02OSGT,HM08}. While we do not discuss the %
operator-space 
Grothendieck inequality here, we remark that in~\cite{RegevV12b} the 
operator-space 
Grothendieck inequality is proved by reducing it to the Pisier-Haagerup noncommutative Grothendieck inequality. Without going into details, we note that this reduction is also algorithmic. Combined with our results, it leads to an algorithmic proof of the 
operator-space 
Grothendieck inequality, together with an accompanying rounding procedure.

In~\cite{RegevV12a}, the last two named authors show that the noncommutative and 
operator-space 
Grothendieck inequalities have consequences in a setting that generalizes that of classical XOR games, called ``quantum XOR games'': in such games, the questions to the players may be quantum states (and the answers are still a single classical bit). 
In~\cite{RegevV12a},
efficient constant-factor approximation algorithms 
are derived %
for the maximum success probability of players in such a game, in three settings: players sharing an arbitrary quantum state, players sharing a maximally entangled state, and players not sharing any entanglement.  Using \expref{Theorem}{thm:alg}, one can efficiently 
compute
strategies achieving the same approximation guarantee. 
These matters are taken up in~\cite{RegevV12a} and will not be  discussed further here.

 \subsection{The noncommutative Grothendieck inequality}\label{sec:gro ineq}

 The natural semidefinite relaxation of~\eqref{eq:grothendieck problem formulation} is
 \begin{equation}\label{eq:gro SDP}
 \sup_{d\in \N} \sup_{x,y\in (S^{d-1})^n}\sum_{i,j=1}^n A_{ij} \langle x_i,y_j\rangle\,,
 \end{equation}
where $S^{d-1}$ is the unit sphere of $\R^d$. 
(The dimension $d$ can be taken to be at most the number of vector variables, $d\leq 2n$.)
Since, being  a semidefinite program (SDP), the quantity appearing in~\eqref{eq:gro SDP} can be computed in polynomial time with arbitrarily good precision (see~\cite{GLS93}), the fact that the Grothendieck optimization problem admits a 
polynomial-time 
constant-factor %
approximation algorithm follows from the following inequality, which is a classical inequality of Grothendieck of major importance to several mathematical disciplines (see Pisier's survey~\cite{PisierGT} and the references therein for much more on this topic; the formulation of the inequality as below is due to Lindenstrauss and Pe{\l}czy\'{n}ski~\cite{Linden68}):
\begin{equation}\label{eq:gro ineq first statement}
 \sup_{d \in \N} \sup_{x,y\in (S^{d-1})^n}\sum_{i,j=1}^n A_{ij} \langle x_i,y_j\rangle\le K_G \sup_{\e,\d\in \{-1,1\}^n}\sum_{i,j=1}^n A_{ij} \e_i\d_j\,.
\end{equation}
Here $K_G\in (0,\infty)$, which is understood to be the infimum over those constants for which~\eqref{eq:gro ineq first statement} holds true for all $n\in \N$ and all $A\in M_n(\R)$, is a universal constant known as the (real) Grothendieck constant. Its exact value remains unknown, the best available bounds~\cite{Ree91,BMMN11} being $1.676<K_G<1.783$. In order to actually find an assignment $\eps,\delta$ to~\eqref{eq:grothendieck problem formulation} that is within a constant factor of the optimum one needs to argue that a proof of~\eqref{eq:gro ineq first statement} can be turned into an efficient rounding algorithm; this is done in~\cite{AN04}.

If one wishes to mimic the above algorithmic success of the Grothendieck inequality in the context of efficient computation of $\opt_\R(\cdot)$, the following natural strategy presents itself: one  should replace real entries of matrices by vectors in $\ell_2$, \ie, consider elements of $M_n(\ell_2)$, and replace the orthogonality constraints underlying the inclusion $U\in\O_n$, namely,
$$
\forall\, i,j\in \{1,\ldots,n\},\quad \sum_{k=1}^n U_{ik}U_{jk}=\sum_{k=1}^n U_{ki}U_{kj}=\d_{ij}\,,
$$
by the corresponding constraints using scalar product. Specifically, given an $n\times n$ vector-valued matrix $X\in M_n(\ell_2)$ define two real matrices $XX^*, X^*X\in M_n(\R)$ by
\begin{equation}\label{eq:vector abs matrix}
\forall\, i,j\in \{1,\ldots,n\},\quad (XX^*)_{ij}\eqdef \sum_{k=1}^n \langle X_{ik},X_{jk}\rangle\qquad\mathrm{and}\qquad (X^*X)_{ij}\eqdef \sum_{k=1}^n \langle X_{ki},X_{kj}\rangle\,,
\end{equation}
and let the set of $d$-dimensional vector-valued orthogonal matrices be given by
\begin{equation}\label{eq:def: vector orthogonal group}
\O_n(\R^d)\eqdef\left\{X\in M_n(\R^d):\ XX^*=X^*X=I\right\}.
\end{equation}
One then considers the following quantity associated to every $M\in M_n(M_n(\R))$:
\begin{equation}\label{eq:def sdpR}
\SDP_\R(M)\eqdef \sup_{d\in\N}\sup_{X,Y\in \O_n(\R^d)} \sum_{i,j,k,l=1}^n M_{ijkl} \left\langle X_{ij},Y_{kl}\right\rangle.
\end{equation}
Since the constraints that underlie the inclusion $X,Y\in \O_n(\R^d)$ are linear equalities in the pairwise scalar products of the entries of $X$ and $Y$, the quantity $\SDP_\R(M)$ is a semidefinite program and can therefore be computed in polynomial time with arbitrarily good precision. One would therefore aim to prove the following noncommutative variant of the Grothendieck inequality~\eqref{eq:gro ineq first statement}:
\begin{equation}\label{eq:real noncommutative gro}
\forall\, n\in \N,\ \forall\, M\in M_n(M_n(\R)),\quad \SDP_\R(M)\le O(1)\cdot \opt_\R(M).
\end{equation}
The term ``noncommutative'' refers here to the fact that $\opt_\R(M)$ is an optimization problem over the noncommutative group $\O_n$, while the classical Grothendieck inequality addresses an optimization problem over the commutative group $\{-1,1\}^n$. In the same vein, noncommutativity is manifested by the fact that the classical Grothendieck inequality corresponds to the special case of ``diagonal'' $4$-tensors $M\in M_n(M_n(\R))$, \ie, those that satisfy $M_{ijkl}=0$ whenever $i\neq j$ or $k\neq l$.

Grothendieck conjectured~\cite{Gro53} the validity of~\eqref{eq:real noncommutative gro} in 1953, a conjecture that remained open until its 1978 affirmative solution by Pisier~\cite{Pisier78NCGT}. A simpler, yet still highly nontrivial proof of the noncommutative Grothendieck inequality~\eqref{eq:real noncommutative gro} was obtained by Kaijser~\cite{Kaijser83}. In \expref{Section}{sec:real} we design a rounding algorithm corresponding to~\eqref{eq:real noncommutative gro} based on Kaijser's approach. This settles the case of real $4$-tensors of \expref{Theorem}{thm:alg}, albeit with worse approximation guarantee than the one claimed in \expref{Remark}{rem:constants}. The algorithm modeled on Kaijser's proof is interesting in its own right, and seems to be versatile and applicable to other problems, such as possible non-bipartite extensions of the noncommutative  Grothendieck inequality in the spirit of~\cite{AMMN06}; we shall not pursue this direction here.

A better approximation guarantee, and arguably an even more striking rounding algorithm, arises from the work of Haagerup~\cite{Haagerup85NCGT} on the complex version of~\eqref{eq:real noncommutative gro}. In \expref{Section}{sec:realrounding}
 we show how the real case of \expref{Theorem}{thm:alg} follows formally from our results on its complex counterpart, so from now on we focus our attention on the complex case.

\subsubsection{The complex case}\label{sec:haa ineq}

In what follows we let $S^{d-1}_\C$ denote the unit sphere of $\C^d$ (thus $S^0_\C$ can be identified with the unit circle $S^1\subseteq \R^2$). The classical complex Grothendieck inequality~\cite{Gro53,Linden68} asserts that there exists $K\in (0,\infty)$ such that
\begin{equation}\label{eq:complex gro def}
\forall\, n\in \N,\ \forall\, A\in M_n(\C),\quad \sup_{x,y\in (S^{2n-1}_\C)^n}\Big|\sum_{i,j=1}^n A_{ij}\langle x_i,y_j\rangle\Big|\le O(1)\sup_{\alpha,\beta\in (S_\C^0)^n}\Big|\sum_{i,j=1}^n A_{ij} \alpha_i\overline{\beta_j}\Big|\,.
\end{equation}
Let $K_G^\C$ denote the infimum over those $K\in (0,\infty)$ for which~\eqref{eq:complex gro def} holds true. The exact value of $K_G^\C$ remains unknown, the best  available bounds being $1.338<K_G^\C<1.4049$ (the left inequality is due to unpublished work of Davie, and the right one is due to Haagerup~\cite{Haa87}).

For $M\in M_n(M_n(\C))$ we define
\begin{equation}\label{eq:SDPC}
\SDP_\C(M)\eqdef \sup_{d\in \N} \sup_{X,Y\in \U_n(\C^{d})} \Big|\sum_{i,j,k,l=1}^n M_{ijkl} \left\langle X_{ij},Y_{kl}\right\rangle \Big|\,,
\end{equation}
where analogously to~\eqref{eq:def: vector orthogonal group} we set
$$
\U_n(\C^{d})\eqdef\left\{X\in M_n(\C^d):\ XX^*=X^*X=I\right\}\,.
$$
Here for $X\in M_n(\C^d)$ the complex matrices $XX^*,X^*X\in M_n(\C)$ are defined exactly as in~\eqref{eq:vector abs matrix}, with the scalar product being the complex scalar product. Haagerup proved~\cite{Haagerup85NCGT} that
\begin{equation}\label{eq:complex sharp noncommutative gro}
\forall\, n\in \N,\ \forall\, M\in M_n(M_n(\C)),\quad \SDP_\C(M)\le 2\cdot \opt_\C(M)\,.
\end{equation}
Our main algorithm is an efficient rounding scheme corresponding to inequality~\eqref{eq:complex sharp noncommutative gro}.
The constant $2$ in~\eqref{eq:complex sharp noncommutative gro} is sharp, as shown in~\cite{HI95} (see also~\cite[Sec.~12]{PisierGT}).

We note that the noncommutative Grothendieck inequality, as it usually appears in the literature, involves a slightly more relaxed semidefinite program.
In order to describe it, we first remark that instead of maximizing over $X,Y\in \U_n(\C^{d})$ in~\eqref{eq:SDPC} we could equivalently maximize over $X,Y\in M_n(\C^d)$ satisfying $XX^*,X^*X,YY^*,Y^*Y \le I$, which is the same as the requirement $\|XX^*\|,\|X^*X\|,\|YY^*\|,\|Y^*Y\|\le 1$, where here and in what follows $\|\cdot\|$ denotes the operator norm of matrices. This fact is made formal in \expref{Lemma}{lem:equiv-sdp} below. By relaxing the constraints to $\|XX^*\|+\|X^*X\|\le 2$ and $\|YY^*\|+\|Y^*Y\|\le 2$, we obtain the following quantity, which can be shown to still be a semidefinite program: 
\begin{equation}\label{eq:ncnorm}
\nc{M} \eqdef
  \sup_{d\in \N}  \sup_{\substack{X,Y\in M_n(\C^{d})\\ \|XX^*\|+\|X^*X\|\le 2\\\|YY^*\|+\|Y^*Y\|\le 2}}\Big|\sum_{i,j,k,l=1}^n M_{ijkl} \left\langle X_{ij},Y_{kl}\right\rangle\Big|\,.
\end{equation}
 Clearly $\nc{M} \ge \SDP_\C(M)$ for all $M\in M_n(M_n(\C))$. Haagerup proved~\cite{Haagerup85NCGT}  that the following stronger inequality holds true for all $n\in\N$ and $ M\in M_n(M_n(\C))$:
\begin{equation}\label{eq:vect}
\nc{M} \le 2\cdot \opt_\C(M)\,.
\end{equation}
As our main focus is algorithmic, in the following discussion we will establish a rounding algorithm for the tightest relaxation~\eqref{eq:complex sharp noncommutative gro}. In \expref{Section}{sec:round nc} we show that the same rounding procedure can be used to obtain an algorithmic analogue of~\eqref{eq:vect} as well.

\subsubsection{The rounding algorithm}\label{sec:rounding intro}

Our main algorithm is an efficient rounding scheme corresponding to~\eqref{eq:complex sharp noncommutative gro}. In order to describe it, we first introduce the following notation. Let $\varphi:\R\to \R^+$ be given by
\begin{equation}\label{eq:def varphi}
 \varphi(t) \eqdef \frac{1}{2} \text{sech}\Big( \frac{\pi}{2} t\Big)\,=\, \frac{1}{e^{\pi t/2}+e^{-\pi t/2} }\,.
 \end{equation}
One computes that $\int_\R\varphi(t)\text{d} t=1$, so $\varphi$ is a density of a probability measure $\mu$ on $\R$, known as the \emph{hyperbolic secant distribution}. By~\cite[Sec.~23.11]{Joh95} we have
\begin{equation}\label{eq:characteristic}
\forall\, a\in (0,\infty),\quad \int_\R  a^{it}\varphi(t)\text{d}t =\frac{2a}{1+a^2}\,.
\end{equation}
It is possible to efficiently sample from $\mu$ using standard techniques; see, \eg,~\cite[Ch.~IX.7]{DevroyeBook}.

In what follows, given $X\in
M_n(\C^d)$ and $z\in \C^d$ we denote by $\langle X,z\rangle \in M_n(\C)$ the
matrix whose entries are $\langle X,z\rangle_{jk}=\langle X_{jk},z\rangle$.

\begin{figure}[ht]
\begin{protocol*}{Rounding procedure}
\begin{step}
\item Let $X,Y\in M_n(\C^d)$ be given as input. Choose $z\in \{1,-1,i,-i\}^d$ uniformly at random, and sample $t\in \R$ according to the hyperbolic secant distribution $\mu$.
\item Set $\displaystyle X_z \eqdef \frac{1}{\sqrt{2}}\langle
  X,z\rangle\in M_n(\C)$ and $\displaystyle Y_z\eqdef \frac{1}{\sqrt{2}}\left\langle Y,z\right\rangle\in M_n(\C)$.
\item Output the pair of matrices \[
(A,B)=(A(z,t),B(z,t))\eqdef(U_z |X_z|^{it},V_z|Y_z|^{-it})\in
\U_n\times \U_n\,,
\]
where 
$X_z = U_z|X_z|$ and $Y_z = V_z|Y_z|$ are the polar decompositions of $\Xz$ and $\Yz$, respectively.
\end{step}
\end{protocol*}
\caption{The rounding algorithm takes as input a pair of vector-valued matrices $X,Y\in M_n(\C^{d})$. It 
outputs %
two matrices $A,B\in \U_n(\C)$.}
\label{fig:rounding}
\end{figure}

\begin{theorem}\label{thm:roundingSDP} Fix $n,d\in \N$ and $\e\in (0,1)$. Suppose that $M\in M_n(M_n(\C))$ and that  $X,Y\in \U_n(\C^{d})$ are such that
\begin{equation}\label{eq:maximizer assumption}
\Big|\sum_{i,j,k,l=1}^n M_{ijkl} \langle X_{ij}, Y_{kl}\rangle \Big| \,\geq\, (1-\eps)\SDP_\C(M)\,,
\end{equation}
where $\SDP_\C(M)$ is given in~\eqref{eq:SDPC}. Then the rounding procedure described in \expref{Figure}{fig:rounding} 
outputs %
a pair of matrices $A,B\in \U_n$ such that
\begin{equation}\label{eq:goal expectation AB-SDP}
\E\Big[\,\Big|\sum_{i,j,k,l=1}^n M_{ijkl} \,A_{ij} \overline{B_{kl}}\Big|\,\Big] \,\geq\, \Big(\frac{1}{2}-\eps\Big)\SDP_\C(M)\,.
\end{equation}
Moreover, rounding can be performed in time polynomial in $n$ and $\log(1/\eps)$, and can be derandomized in time $\poly(n,1/\eps)$.
\end{theorem}

While the rounding procedure of \expref{Figure}{fig:rounding} and the proof of \expref{Theorem}{thm:roundingSDP} (contained in \expref{Section}{sec:proof} below) appear to be different from Haagerup's original proof of~\eqref{eq:vect} in~\cite{Haagerup85NCGT}, we derived them using Haagerup's ideas. One source of difference arises from changes that we introduced in order to work with the quantity $\SDP_\C(M)$, while Haagerup's argument treats the quantity $\nc{M}$. A second source of difference is that Haagerup's proof  of~\eqref{eq:vect}  is rather indirect and nonconstructive, while it is crucial to the algorithmic applications that were already mentioned in \expref{Section}{sec:apps intro} for us to formulate a polynomial-time rounding procedure. Specifically, Haagerup establishes the \emph{dual} formulation of~\eqref{eq:vect}, through a repeated use of duality, and he uses a bootstrapping argument that relies on nonconstructive tools from complex analysis. The third step in \expref{Figure}{fig:rounding} originates from Haagerup's complex-analytic considerations. Readers who are accustomed to semidefinite rounding techniques will immediately notice that this step is unusual; we give intuition for it in \expref{Section}{sec:intuition} below, focusing for simplicity on applying the rounding procedure to vectors rather than matrices (\ie, the more familiar setting of the classical Grothendieck inequality).

\subsubsection{An intuitive description of the rounding procedure in the commutative case}\label{sec:intuition}

Consider the effect of the rounding procedure in the commutative case, \ie, when $X,Y\in M_n(\C^d)$ are diagonal matrices. Let the diagonals of $X,Y$ be $x,y\in (\C^d)^n$, respectively. The first step consists in performing a random projection: for $j\in \{1,\ldots,n\}$ let $\alpha_j =\left\langle x_j,z\right\rangle/\sqrt{2}\in \C$ and $\beta_j =\left\langle y_j,z\right\rangle/\sqrt{2}\in \C$, where $z$ is chosen uniformly at random from $\{1,-1,i,-i\}^n$ (alternatively, with minor modifications to the proof one may choose i.\,i.\,d.\ $z_j$ uniformly from the unit circle, as was done by Haagerup~\cite{Haagerup85NCGT}, or use standard complex Gaussians). This step results in sequences of complex numbers whose pairwise products $\alpha_k\overline{\beta_j}$, in expectation, exactly reproduce the pairwise scalar products $\langle x_k,y_j\rangle/2$. However, in general the resulting complex numbers $\alpha_k$ and $\beta_j$ may have modulus larger than $1$. Extending the ``sign'' rounding performed in, say, the Goemans-Williamson algorithm for \textsc{MaxCut}~\cite{GW95} to the complex domain, one could then round each $\alpha_k$ and $\beta_j$ independently by simply replacing them by their respective complex phase.

The procedure that we consider differs from this standard practice by taking into account potential information contained in the modulus of the random complex numbers $\alpha_k,\beta_j$. Writing in polar coordinates $\alpha_k=r_k e^{i\theta_k}$ and $\beta_j=s_je^{i\phi_j}$ we sample a real $t$ according to a specific distribution (the hyperbolic secant distribution $\mu$), and round each $\alpha_k$ and each $\beta_j$ to
$$
a_k \eqdef e^{i(\theta_k+t\log r_k)}\in S_\C^0\,, \qquad\mathrm{and}\qquad b_j \eqdef e^{i(\phi_j-t\log s_j)} \in S_\C^0\,,
$$
respectively.
Observe that this step performs a \emph{correlated} rounding: the parameter $t$ is the same for all $j,k\in \{1,\ldots,n\}$.

The proof presented in~\cite{Haagerup85NCGT} uses the maximum modulus principle to show the \emph{existence} of a real $t$ for which $a_k,b_j$ as defined above provide a good assignment. Intuition for the existence of such a good $t$ can be given as follows. Varying $t$ along the real line corresponds to rotating the phases of the complex numbers $\alpha_j,\beta_k$ at a speed proportional to the logarithm of their modulus: elements with very small modulus vary very fast, those with modulus $1$ are left unchanged, and elements with relatively large modulus are again varied at (logarithmically) increasing speeds. This means that the rounding procedure takes into account the fact that an element with modulus away from $1$ is a ``miss'': that particular element's phase is probably irrelevant, and should be changed. However, elements with modulus close to $1$ are ``good'': their phase can be kept essentially unchanged.

We identify a specific distribution $\mu$ such that a random $t$ distributed according to $\mu$ is good, in expectation. This results in a variation on the usual ``sign'' rounding technique: instead of directly keeping the phases obtained in the initial step of random projection, they are synchronously rotated for a random time $t$, at speeds depending on the associated moduli, resulting in a provably good pair of sequences $a_k,b_j$ of complex numbers with modulus $1$.

\paragraph{Roadmap:} In \expref{Section}{sec:proof} we prove \expref{Theorem}{thm:roundingSDP} both as stated in \expref{Section}{sec:rounding intro} and in a form based on~\eqref{eq:vect}. The real case as well as a closely related Hermitian case are treated next, first in \expref{Section}{sec:realrounding} as a corollary of \expref{Theorem}{thm:roundingSDP}, and then using an alternative direct rounding procedure in \expref{Section}{sec:real}.
\expref{Section}{sec:applications} presents the applications that were outlined in \expref{Section}{sec:apps intro}.

\section{Proof of \expref{Theorem}{thm:roundingSDP}}\label{sec:proof}

In this section we prove \expref{Theorem}{thm:roundingSDP}. The rounding procedure described in \expref{Figure}{fig:rounding} is analyzed in \expref{Section}{sec:analysis}, while the derandomized version is presented in \expref{Section}{sec:derandomized}. The efficiency of the procedure is clear; we refer to \expref{Section}{sec:derandomized} for a discussion on how to discretize the choice of $t$. In \expref{Section}{sec:round nc} we show how the analysis can be modified to the case of $\nc{M}$ and~\eqref{eq:vect}.

In what follows, it will be convenient to use the following notation. Given $M\in M_n(M_n(\C))$ and $X,Y\in M_n(\C^d)$, define
\begin{equation}\label{eq:def sesqui}
M(X,Y)\eqdef \sum_{i,j,k,l=1}^n M_{ijkl} \left\langle X_{ij},Y_{kl}\right\rangle\in \C\,.
\end{equation}
Thus $M(\cdot,\cdot)$ is a sesquilinear form on $M_n(\C^d)\times M_n(\C^d)$, \ie, $M(\alpha X,\beta Y)=\alpha\overline{\beta} M(X,Y)$ for all $X,Y\in M_n(\C^d)$ and $\alpha,\beta\in \C$. Observe that if $A,B\in M_n(\C)$ then
\begin{equation}\label{eq:write in terms of tensors}
M(A,B)=\sum_{i,j,k,l=1}^n M_{ijkl} A_{ij}\overline{B_{kl}}= \sum_{i,j,k,l=1}^n M_{ijkl} \left(A\otimes \overline{B}\right)_{(ij),(kl)}.
\end{equation}

\subsection{Analysis of the rounding procedure}\label{sec:analysis}

\begin{proof}[Proof of~\eqref{eq:goal expectation AB-SDP}]
Let $X,Y\in \U_n(\C^d)$ be vector-valued matrices satisfying~\eqref{eq:maximizer assumption}. Let $z\in \{1,-1,i,-i\}^d$ be chosen uniformly at random, and
$$\Xz \eqdef \frac{1}{\sqrt{2}}\left\langle X,z\right\rangle\qquad\text{and} \qquad \Yz \eqdef \frac{1}{\sqrt{2}}\left\langle Y,z\right\rangle$$
be random variables taking values in $M_n(\C)$ defined as in the second step of the rounding procedure (see \expref{Figure}{fig:rounding}). Then,
\begin{equation}\label{eq:an-1}
\E_{z} \big[\, M( \Xz,  \Yz )\,\big] = \frac{1}{2}\E_{z}\Big[\,\sum_{r,s=1}^d \overline{z_r} {z_s} \sum_{i,j,k,l=1}^n  M_{ijkl} (X_{ij})_r \overline{(Y_{kl})_{s}}\,\Big]
=\frac{1}{2} M ( X, Y )\,,
\end{equation}
where we used the fact that $\E[\overline{z_r}z_s]=\delta_{rs}$ for every $r,s\in\{1,\ldots,d\}$.

Observe that~\eqref{eq:characteristic} implies that
$$\forall\, a\in (0,\infty)\,,\quad  \E_{t} [ a^{it} ] = 2a - \E_{t} [ a^{2+it} ]\,.$$
Recall that the output of our rounding scheme as described in \expref{Figure}{fig:rounding} is $A=U_z|\Xz|^{it}$ and $B=V_z|\Yz|^{-it}$, where $\Xz = U_z|\Xz|$ and $\Yz=V_z|\Yz|$ are the polar decompositions of $\Xz$ and $\Yz$, respectively. Applying the above identity to the the (positive) eigenvalues of $|\Xz|\otimes |\Yz|$, we deduce the matrix equality
 \begin{eqnarray}\label{eq:mu identity bootstrap}
\E_t\left(A\otimes \overline{B}\right) &=&\nonumber   \E_{t} \left[\, \left(U_z|\Xz|^{it}\right)\otimes \left(\overline{V_z}|\Yz|^{it}\right) \,\right] \\
&=& \nonumber \big(U_z\otimes \overline{V_z}\big)\cdot\Big( \E_{t} \left[\, \left(|\Xz|\otimes |\Yz|\right)^{it} \,\right]\Big) \\
&=&2U_z|\Xz|\otimes \overline{V_z}|\Yz| - \E_{t} \left[\, \left(U_z|\Xz|^{2+it}\right)\otimes \left(\overline{V_z}|\Yz|^{2+it}\right) \,\right]\nonumber\\
&=&2\Xz\otimes \overline{\Yz} - \E_{t} \left[\, \left(U_z|\Xz|^{2+it}\right)\otimes \left(\overline{V_z|\Yz|^{2-it}}\right) \,\right].
  \end{eqnarray}

  It follows from~\eqref{eq:def sesqui}, \eqref{eq:write in terms of tensors}, \eqref{eq:an-1} and~\eqref{eq:mu identity bootstrap} that
\begin{equation}\label{eq:an-2}
\E_{z, t} \big[\, M(A, B) \,\big]= M(X, Y) - \E_{z,t} \big[\, M\big(U_z|\Xz|^{2+it}, V_z|\Yz|^{2-it}\big) \,\big].
\end{equation}
Our goal from now on is to bound the second, ``error'' term on the right-hand side of~\eqref{eq:an-2}. Specifically, the rest of the proof is devoted to showing that for any fixed $t \in \R$ we have
\begin{equation}\label{eq:half sdp statement}
  \left|\E_{z} \big[\, M\big(U_z|\Xz|^{2+it}, V_z|\Yz|^{2-it}\big) \,\big]\right| \le \frac12\SDP_\C(M)\,.
\end{equation}
 Once established, the estimate~\eqref{eq:half sdp statement} completes the proof of the desired expectation bound~\eqref{eq:goal expectation AB-SDP} since \begin{equation*} \E_{z, t} \big[\, |M(A, B)| \,\big]\stackrel{\eqref{eq:an-2}\wedge\eqref{eq:half sdp statement}}{\ge} M(X,Y)-\frac12 \SDP_\C(M)\stackrel{\eqref{eq:maximizer
assumption}}{\ge}\left(\frac12 -\e\right)\SDP_\C(M)\,.
\end{equation*}
So, for the rest of the proof, fix some $t \in \R$. As a first step towards~\eqref{eq:half sdp statement} we state the following claim.

\begin{claim}\label{claim:haag-id} Let $W\in M_n(\C^d)$ be a vector-valued matrix, and for every $r\in \{1,\ldots,d\}$ define $W_r\in M_n(\C)$ by $(W_r)_{ij}=(W_{ij})_r$. Let $z\in \{1,-1,i,-i\}^d$ be chosen uniformly at random. Writing $W_z=\langle W,z\rangle\in M_n(\C)$, we have
\begin{align}
\E_z\left[(W_zW_z^*)^2\right]&=(WW^*)^2 + \sum_{r=1}^d W_r(W^*W-W_r^*W_r)W_r^*\,,\label{eq:identity version-1}\\
\E_z\left[(W_z^*W_z)^2\right]&=(W^*W)^2 + \sum_{r=1}^d W_r^*(WW^*-W_rW_r^*)W_r\,.\label{eq:identity version-2}
\end{align}
\end{claim}

\begin{proof}
By definition $W_z=\sum_{r=1}^d \overline{z_r} {W_r}$, and
recalling~\eqref{eq:vector abs matrix} we have
\[
WW^*=\sum_{r=1}^d W_rW_r^*\qquad\text{and}\qquad W^*W=\sum_{r=1}^d
  W_r^*W_r\,.
\]
Consequently,
\begin{align}\label{eq:for unitary case}
\nonumber\E_z\left[(W_zW_z^*)^2\right]&=\E_z\Big[\, \sum_{p,q,r,s=1}^d\overline{z_p}{z_q}\overline{z_r}{z_s}{W_p}W_q^*{W_r}
W_s^*\, \Big]\\ \nonumber
&= \sum_{p=1}^d W_pW_p^*W_p
W_p^*+\sum_{\substack{p,q\in \{1,\ldots,d\}\\p\neq q}}\left(W_pW_p^*W_q
W_q^*+W_pW_q^*W_q
W_p^*\right)\\\nonumber
&= \sum_{p,q=1}^d W_pW_p^*W_q
W_q^*+\sum_{p,q=1}^dW_pW_q^*W_q
W_p^*- \sum_{p=1}^dW_pW_p^*W_pW_p^* \\\nonumber
&= \Big(\sum_{p=1}^d W_pW_p^*\Big)^2+\sum_{p=1}^d W_p\Big(\sum_{q=1}^dW_q^*W_q\Big)W_p^* - \sum_{p=1}^dW_pW_p^*W_pW_p^*\,,
\end{align}
proving~\eqref{eq:identity version-1}. A similar calculation yields~\eqref{eq:identity version-2}.
\end{proof}

Now, for every $t\in \R$ define two vector-valued matrices $$F(t),G(t)\in M_n\left(\C^{\{1,-1,i,-i\}^d}\right)$$ by setting for every $j,k\in \{1,\ldots,n\}$ and $z\in \{1,-1,i,-i\}^d$,
\begin{equation}\label{eq:defFG}
(F(t)_{jk})_z\eqdef \frac{1}{2^d}\left(U_z|\Xz|^{2+it}\right)_{jk} \qquad \text{and}\qquad (G(t)_{jk})_z\eqdef \frac{1}{2^d}\left(V_z|\Yz|^{2-it}\right)_{jk}.
\end{equation}
Thus,
\begin{equation}\label{eq:MFG}
M(F(t),G(t))=\frac{1}{4^d}\sum_{z\in \{1,-1,i,-i\}^d} M\big(U_z|\Xz|^{2+it}, V_z|\Yz|^{2-it}\big) =\E_{z} \big[\, M\big(U_z|\Xz|^{2+it}, V_z|\Yz|^{2-it}\big) \,\big].
\end{equation}
Moreover, recalling that $X_z=U_z|X_z|$ is the polar decomposition of $X_z$, we have
\begin{equation}\label{eq:F identity}
F(t)F(t)^*=\frac{1}{4^d}\sum_{z\in \{1,-1,i,-i\}^d} U_z |X_z|^4 U_z^*=\E_z\left[U_z |X_z|^4 U_z^*\right]=\E_z\left[\left(X_zX_z^*\right)^2\right].
\end{equation}
Similarly $F(t)^*F(t)=\E_z\left[\left(X_z^*X_z\right)^2\right]$, so that an application of \expref{Claim}{claim:haag-id} with $W=\frac{1}{\sqrt{2}}X$ yields, using $XX^*=X^*X=I$ since $X\in \U_n(\C^d)$,
\begin{equation}\label{eq:identity F}
F(t)F(t)^*+\frac14 \sum_{r=1}^d X_rX_r^*X_rX_r^*=F(t)^*F(t)+\frac14 \sum_{r=1}^d X_r^*X_rX_r^*X_r=\frac{1}{2}I.
\end{equation}
Analogously,
\begin{equation}\label{eq:identity G}
G(t)G(t)^*+\frac14 \sum_{r=1}^d Y_rY_r^*Y_rY_r^*=G(t)^*G(t)+\frac14 \sum_{r=1}^d Y_r^*Y_rY_r^*Y_r=\frac{1}{2}I.
\end{equation}
Using that each term $X_r^*X_rX_r^*X_r$ and $Y_r^*Y_rY_r^*Y_r$ is positive semidefinite, the two equations above imply that $F(t),G(t)$ satisfy the norm bounds
\begin{equation}\label{eq:FG-maxbound}
\max\big\{\|F(t)F(t)^*\|,\,\|F(t)^*F(t)\|,\,\|G(t)G(t)^*\|,\,\|G(t)^*G(t)\|\big\}\leq \frac{1}{2}\,.
\end{equation}
 As shown in \expref{Lemma}{lem:equiv-sdp} below,~\eqref{eq:FG-maxbound} implies that there exists a pair of vector-valued matrices $$R(t),S(t)\in\U_n(\C^{d+2n^2})$$ such that \begin{equation}\label{eq:MRS}
 M(R(t),S(t))=M(\sqrt{2}F(t),\sqrt{2}G(t))\,.
 \end{equation}
  (This fact can also be derived directly using~\eqref{eq:identity F} and~\eqref{eq:identity G}.) Recalling the definition of $\SDP_\C(M)$ in~\eqref{eq:SDPC}, it follows that for every $t\in \R$,
\begin{equation}\label{eq:half sdp}
 \left|\E_{z} \big[\, M\big(U_z|\Xz|^{2+it}, V_z|\Yz|^{2-it}\big) \,\big]\right|\stackrel{\eqref{eq:MFG}}{=}|M(F(t),G(t))|\stackrel{\eqref{eq:MRS}}{=}\frac12|M(R(t),S(t))|\le \frac12\SDP_\C(M)\,,
\end{equation}
completing the proof of~\eqref{eq:half sdp statement}.
\end{proof}

\begin{lemma}\label{lem:equiv-sdp}
Let $X,Y\in M_n(\C^d)$ be such that $\max(\|X^*X\|,\|XX^*\|,\|Y^*Y\|,\|YY^*\|)\leq 1$. Then there exist $R,S\in \U_n(\C^{d+2n^2})$ such that for every $M\in M_n(M_n(\C))$ we have $M(R,S) = M(X,Y)$. %
Moreover, $R$ and $S$ can be computed from $X$ and $Y$ in time $\poly(n,d)$.
\end{lemma}

\begin{proof}
Let $A = I - XX^*$ and $B = I-X^*X$, and note that $A,B\geq 0$ and
$\Tr(A)=\Tr(B)$. Write the spectral decompositions of $A$ and $B$ as
\[
A = \sum_{i=1}^n \lambda_i (u_iu_i^*)\qquad\text{and}\qquad B =
\sum_{j=1}^n \mu_j (v_jv_j^*)\,, %
\]
respectively. Set $\sigma = \sum_{i=1}^n \lambda_i = \sum_{j=1}^n \mu_j$, and define
$$R \eqdef X \oplus \Big(\bigoplus_{i,j=1}^n \sqrt{\frac{\lambda_i\mu_j}{\sigma}}\, (u_i v_j^*)\Big) \oplus \Big( 0_{M_n(\C^{n^2})}\Big)\in M_n(\C^{d}\oplus \C^{n^2}\oplus \C^{n^2})\,.$$
With this definition we have $RR^* = XX^* + A = I$ and $R^*R = X^*X + B = I$, so $R\in\U_n(\C^{d+2n^2})$.
Let $S\in \U_n(\C^{d+2n^2})$ be defined analogously from $Y$, with the last two blocks of $n^2$ coordinates permuted. One checks that $M(R,S) = M(X,Y)$, as required.

Finally, $A,B$, their spectral decomposition, and the resulting $R,S$ can all be computed in time $\poly(n,d)$ from $X,Y$.
\end{proof}

\subsection{Derandomized rounding}\label{sec:derandomized}

Note that we can always assume that $X,Y\in \U_n(\C^d)$, where $d\le 2n^2$. We start by slightly changing the projection step. Define $\Xz'$ to be the projection $\Xz=({1}/{\sqrt{2}})\langle X,z \rangle $, after we replace all singular values of $\Xz$ that are smaller than $\eps$ with $\eps$. Then, writing
$$ 2 \Xz' \otimes \overline{\Yz}' = 2 \Xz \otimes \overline{\Yz} + 2 (\Xz'-\Xz) \otimes \overline{\Yz} + 2 \Xz' \otimes (\overline{\Yz'} - \overline{\Yz})\,,$$
we see that in the analogue of~\eqref{eq:an-2} the first term is at least $M(X,Y)-4\eps d \SDP_\C(M)$. Here we use that $\|\Xz\|, \|\Yz\| \le d$ which follows by the triangle inequality and $X,Y\in\U_n(\C^d)$.
For the second term in~\eqref{eq:an-2}, the previous analysis remains
unchanged, provided we prove an analogue
of~\eqref{eq:FG-maxbound}. Using~\eqref{eq:F identity} and the
analogous equations for $F(t)^*F(t)$, $G(t)G(t)^*$ and $G(t)^*G(t)$, it will suffice to bound four expressions such as
\[
\big\| \E_{z}\big[\, (\Xz'(\Xz')^\dagger)^{2}\,\big] \big\|\,.
\]
One checks that the modification to the rounding we did can only increase this by $\eps^4$ (even for each ${z}$), hence following the previous analysis we get the bound in~\eqref{eq:half sdp} with $\SDP_\C(M)/2$ replaced by $(1+\eps^4) \SDP_\C(M)/2$.

Next, we observe that the coordinates of ${z}$ need not be independent, and it suffices if they are chosen from a four-wise independent distribution. As a result, there are only $\poly(n)$ possible values of ${z}$ and they can be enumerated efficiently. Therefore, we can assume that we have a value ${z}\in \{1,-1,i,-i\}^d$ for which
\begin{align}\label{eq:derand1}
\big| \E_{t} \big[\, M\big(U_z|\Xz'|^{it}, V_z|\Yz'|^{-it}\big) \,\big]\big| \,\geq\, \Big(\frac{1}{2}-\eps\Big)\SDP_\C(M)\,.
\end{align}

Notice that for some universal constant $c>0$, with probability at least $1-\eps$, a sample from the hyperbolic secant distribution is at most $c \log(1/\eps)$ in absolute value. Therefore, denoting the restriction of the hyperbolic secant distribution to the interval $[-c \log(1/\eps), c \log(1/\eps)]$ by $\mu'$, and using the fact that the expression inside the expectation in~\eqref{eq:derand1} is never larger than $\SDP_\C(M)$, for $t'$ distributed according to $\mu'$ we have
$$
\big| \E_{t'} \big[\, M\big(U_z|\Xz'|^{it'}, V_z|\Yz'|^{-it'}\big) \,\big]\big| \,\geq\, \Big(\frac{1}{2}-2\eps\Big)\SDP_\C(M)\,.
$$
Moreover, for any $t,t' \in \R$,
\begin{align}
&\big|M\big( U_z|\Xz'|^{it}, V_z|\Yz'|^{-it}\big) - M\big( U_z|\Xz'|^{it'}, V_z|\Yz'|^{-it'}\big)\big|  \nonumber \\
& \qquad =
\big|M\big( U_z(|\Xz'|^{it}- |\Xz'|^{it'}), V_z|\Yz'|^{-it}\big) +
M\big( U_z|\Xz'|^{it'}, V_z(|\Yz'|^{-it}-|\Yz'|^{-it'})\big)\big|  \nonumber \\
& \qquad \le
\big|M\big( U_z(|\Xz'|^{it}- |\Xz'|^{it'}), V_z|\Yz'|^{-it}\big)\big|+
\big|M\big( U_z|\Xz'|^{it'}, V_z(|\Yz'|^{-it}-|\Yz'|^{-it'})\big)\big|\,.\label{eq:smoothint}
\end{align}
The first absolute value in~\eqref{eq:smoothint} is at most
\begin{align*}
\SDP_\C(M) \cdot \| U_z(|\Xz'|^{it}- |\Xz'|^{it'})\| \cdot \| V_z|\Yz'|^{-it} \| &= \SDP_\C(M) \cdot \| |\Xz'|^{it}- |\Xz'|^{it'}\| \\
&\le \SDP_\C(M)\cdot  \log(\max\{ \|\Xz'\|,\|(\Xz')^{-1}\| \})\cdot |t-t'|\,.
\end{align*}
We have $\|\Xz'\| \le d$ as explained above, and $\|(\Xz')^{-1}\| \le 1/\eps$ by the way our modified rounding procedure was defined.
Similar bounds hold for $\Yz'$.
It therefore suffices to pick $t$ from a grid of size $O(\log(1/\eps) \max(1/\eps^2,d/\eps) )$.

\subsection{The rounding procedure in the case \texorpdfstring{of~\eqref{eq:vect}}{of (\ref{eq:vect})}}\label{sec:round nc}

\expref{Theorem}{thm:roundingSDP} addressed the performance of the rounding procedure described in \expref{Figure}{fig:rounding} with respect to inequality~\eqref{eq:complex sharp noncommutative gro}. Here we prove that this rounding procedure has the same performance with respect to the noncommutative Grothendieck inequality~\eqref{eq:vect} as well. This is the content of the following theorem.

\begin{theorem}\label{thm:rounding} Fix $n,d\in \N$ and $\e\in (0,1)$. Suppose that $M\in M_n(M_n(\C))$ and that  $X,Y\in M_n(\C^d)$ satisfy
\begin{equation}\label{eq:max constraint}
\max\left\{\|XX^*\|+\|X^*X\|, \|YY^*\|+\|Y^*Y\|\right\}\le 2
\end{equation}
and
\begin{equation}\label{eq:maximizer assumption-nc}
\Big|\sum_{i,j,k,l=1}^n M_{ijkl} \langle X_{ij}, Y_{kl}\rangle \Big|
\,\geq\, (1-\eps) \nc{M}\,,
\end{equation}
where $\nc{M}$ was defined in~\eqref{eq:vect}. Then the rounding procedure described in \expref{Figure}{fig:rounding} 
outputs %
a pair of matrices $A,B\in \U_n$ such that
\begin{equation}\label{eq:goal expectation AB}
\E\Big[\,\Big|\sum_{i,j,k,l=1}^n M_{ijkl} \,A_{ij} \overline{B_{kl}}\Big|\,\Big] \,\geq\, \Big(\frac{1}{2}-\eps\Big)\nc{M}\,.
\end{equation}
\end{theorem}

\begin{proof}
We shall explain how to modify the argument presented in \expref{Section}{sec:analysis}, relying on the notation that was introduced there. All we need to do is to replace~\eqref{eq:half sdp} by the assertion
\begin{equation}\label{eq:new assertion}
\forall\, t\in \R,\qquad |M(F(t),G(t))|\le \frac12 \nc{M}\,.
\end{equation}
To this end we use the following corollary of \expref{Claim}{claim:haag-id}, which is a slight variant of~\cite[Lem.~4.1]{Haagerup85NCGT}.

\begin{corollary}\label{claim:haag} Let $W\in M_n(\C^d)$ be a vector-valued matrix and $z\in \{1,-1,i,-i\}^d$ chosen uniformly at random. Writing $W_z=\langle W,z\rangle\in M_n(\C)$, we have
 \begin{equation}\label{eq:bounded by square}
\left\|\E_z\left[(W_zW_z^*)^2\right]\right\|+\left\|\E_z\left[(W_z^*W_z)^2\right]\right\| \le \big(\left\|WW^*\right\|+\left\|W^*W\right\|\big)^2.
 \end{equation}
\end{corollary}

\begin{proof}
Starting from~\eqref{eq:identity version-1} and noting that
$$\Big\|\sum_{r=1}^d W_r(W^*W-W_r^*W_r)W_r^*\Big\| \leq \|W^*W\|\cdot \Big\|\sum_{r=1}^d W_r W_r^*\Big\| =\left\|W^*W\right\|\cdot\left\|WW^*\right\|\,,$$
we obtain the inequality
\begin{equation}\label{eq:to add1}
\left\|\E_z\left[(W_zW_z^*)^2\right]\right\|\le \left\|WW^*\right\|^2+\left\|W^*W\right\|\cdot\left\|WW^*\right\|.
\end{equation}
Similarly, from~\eqref{eq:identity version-2} we get
\begin{equation}\label{eq:to add 2}
\left\|\E_z\left[(W_z^*W_z)^2\right]\right\|\le \left\|W^*W\right\|^2+\left\|WW^*\right\|\cdot\left\|W^*W\right\|.
\end{equation}
By summing~\eqref{eq:to add1} and~\eqref{eq:to add 2} one deduces~\eqref{eq:bounded by square}.
\end{proof}

Combining~\eqref{eq:F identity} with~\eqref{eq:bounded by square} for $W=X/\sqrt{2}$, we have
\begin{equation}\label{eq:new identity F}
\|F(t)F(t)^*\|+\|F(t)^*F(t)\|\le \frac{1}{4}(\|XX^*\|+\|X^*X\|)^2\stackrel{\eqref{eq:max constraint}}{\le} 1\,.
\end{equation}
Analogously,
\begin{equation}\label{eq:new identity G}
\|G(t)G(t)^*\|+\|G(t)^*G(t)\|\le \frac{1}{4}(\|YY^*\|+\|Y^*Y\|)^2\le 1\,.
\end{equation}
Recalling the definition of $\nc{M}$, it follows from~\eqref{eq:new identity F} and~\eqref{eq:new identity G} that for every $t\in \R$,
\begin{equation}\label{eq:half nc}
 \left|\E_{z} \big[\, M\big(U_z|\Xz|^{2+it}, V_z|\Yz|^{2-it}\big) \,\big]\right|\stackrel{\eqref{eq:MFG}}{=}|M(F(t),G(t))|\le \frac12\nc{M}\,.
\end{equation}
Hence,
\begin{equation*}
\E_{z, t} \big[\, |M(A, B)| \,\big]\stackrel{\eqref{eq:an-2}\wedge\eqref{eq:half nc}}{\ge} M(X,Y)-\frac12 \nc{M}\stackrel{\eqref{eq:maximizer assumption-nc}}{\ge}\left(\frac12 -\e\right)\nc{M}\,,
\end{equation*}
completing the proof of the desired expectation bound~\eqref{eq:goal expectation AB}.
\end{proof}

\begin{remark}
The following example, due to Haagerup~\cite{Haagerup85NCGT}, shows
that the factor $2$ approximation guarantee obtained in
\expref{Theorem}{thm:rounding} is optimal: the best constant
in~\eqref{eq:vect} equals $2$. Let $M\in M_n(M_n(\C))$ be given by
$M_{1jk1} = \delta_{jk}$, and $M_{ijkl}=0$ if $(i,l)\neq (1,1)$. A
direct computation shows that $\opt_\C(M)=1$. Define $X,Y\in
M_{n}(\C^n)$ by $X_{1j} = Y_{j1} = \sqrt{2/(n+1)}\ e_j \in \C^n$ for
$j\in \{1,\ldots,n\}$ and all other entries of $X$ and $Y$  vanish
(here $e_1,\ldots,e_n$ is the standard basis of $\C^n$). Using these
two vector-valued matrices one shows that $\nc{M} \ge 2n/(n+1)$. Recall that in the introduction we mentioned that it was shown in~\cite{HI95} that the best constant  in the weaker inequality~\eqref{eq:complex sharp noncommutative gro} is also $2$, but the example exhibiting this stronger fact is more involved.
\end{remark}

\section{The real and Hermitian cases}\label{sec:realrounding}

The $n\times n$ Hermitian matrices are denoted $H_n$. A $4$-tensor $M\in M_n(M_n(\C))\cong M_n(\C)\otimes M_n(\C)$ is said to be Hermitian if $M_{ijkl}=\overline{M_{jilk}}$ for all $i,j,k,l\in \{1,\ldots,n\}$. Investigating the noncommutative Grothendieck inequality in the setting of Hermitian $M$ is most natural in applications to quantum information, while problems in real optimization as described in the introduction lead to real  $M\in M_n(M_n(\R))$. These special cases are treated in this section.

Consider the following Hermitian analogue of the quantity $\opt_\C(M)$:
$$
\opt^*_\C(M)\eqdef \sup_{\substack{A,B\in H_n\\ \|A\|,\|B\|\le 1}}\Big|\sum_{i,j,k,l=1}^n M_{ijkl}A_{ij}\overline{B_{kl}}\Big|\,.
$$
Note that the convex hull of $\U_n$ consists of all the matrices $A\in M_n(\C)$ with $\|A\|\le 1$, so by convexity for every $M\in M_n(M_n(\C))$ we have
\begin{equation}\label{eq:pass to convex hull C}
\opt_\C(M)= \sup_{\substack{A,B\in M_n(\C)\\\|A\|,\|B\|\le 1}}\Big|\sum_{i,j,k,l=1}^n M_{ijkl}A_{ij}\overline{B_{kl}}\Big|\,.
\end{equation}
This explains why $\opt^*_\C(M)$ should indeed be viewed as a Hermitian analogue of $\opt_\C(M)$. The real analogue of~\eqref{eq:pass to convex hull C} is that, due to the fact that the convex hull of $\O_n$ consists of all the matrices $A\in M_n(\R)$ with $\|A\|\le 1$, for every $M\in M_n(M_n(\R))$ we have
\begin{equation}\label{eq:pass to convex hull R}
\opt_\R(M)=\sup_{\substack{A,B\in M_n(\R)\\\|A\|,\|B\|\le 1}}\Big|\sum_{i,j,k,l=1}^n M_{ijkl}A_{ij}B_{kl}\Big|\,.
\end{equation}

The following theorem establishes an algorithmic equivalence between the problems of approximating either of these two quantities.

\begin{theorem}\label{thm:hermitian real equivalence}
For every $K\in [1,\infty)$ the following two assertions are equivalent.
\begin{enumerate}
\item\label{part:hermite-1} There exists a 
polynomial-time 
algorithm $\alg^*$ that takes as input a Hermitian $M\in M_n(M_n(\C))$ and 
outputs %
$A,B\in H_n$ with $\max\{\|A\|,\|B\|\}\le 1$ and
$\opt_\C^*(M)\le K|M(A,B)|$.
\item\label{part:hermite-2} There exists a 
polynomial-time 
algorithm $\alg$ that takes as input $M\in M_n(M_n(\R))$ and 
outputs %
$U,V\in \O_n$ such that $\opt_\R(M)\le K M(U,V)$.
\end{enumerate}
\end{theorem}

In \expref{Section}{sec:hermitian} we show that for every
$K>2\sqrt{2}$ there exists an algorithm $\alg^*$ as in
\expref{assertion}{part:hermite-1} of \expref{Theorem}{thm:hermitian
  real equivalence}.
 Consequently, we obtain the algorithm for
computing $\opt_\R(M)$ whose existence was claimed in
\expref{Theorem}{thm:alg}. The implication
\ref{part:hermite-1}~$\implies$~\ref{part:hermite-2} of \expref{Theorem}{thm:hermitian real equivalence} is the only part of \expref{Theorem}{thm:hermitian real equivalence} that will be used in this article; the reverse direction \ref{part:hermite-2}~$\implies$~\ref{part:hermite-1} is included here for completeness. Both directions of \expref{Theorem}{thm:hermitian real equivalence} are proved in \expref{Section}{sec:real-herm-equiv}.

\subsection{Two-dimensional rounding}\label{sec:twodim}

In this section we give an algorithmic version of Krivine's proof~\cite{Krivine:79a} that the $2$-dimensional real Grothendieck constant satisfies $K_G(2) \leq \sqrt{2}$. The following theorem is implicit in the proof of~\cite[Thm.~1]{Krivine:79a}.

\begin{theorem}[Krivine]\label{thm:krivine}
Let $g:\R\to \R$ be defined by $g(x)=\mathrm{sign}(\cos(x))$, and let $f:[0,\pi/2)\to \R$ be given by
$$
 f(t) \eqdef \begin{cases} 1&\text{ if\  $0\leq t\leq \frac{\pi}{4}$,}\\ \frac{6}{\pi} \Big(\frac{\pi}{2}-t\Big)-\frac{1}{2}\Big(\frac{4}{\pi}\Big)^3\Big(\frac{\pi}{2}-t\Big)^3&\text{ if\  $\frac{\pi}{4} \leq t < \frac{\pi}{2}$.}\end{cases}
$$
Extend $f$ to a function defined on all of $\R$ by requiring that it is even and $f(x+\pi)=-f(x)$ for all $x\in \R$. There exists a sequence $\{b_{2\ell+1}\}_{\ell=0}^\infty \in \R^\N$ such that for every $L\in \N$ the numbers $\{b_0,\ldots,b_{2L+1}\}$ can be computed in $\poly(L)$ time, $\sum_{\ell= L+1}^\infty |b_{2\ell+1}| \leq C/L$ for some universal constant $C$, $\sum_{\ell=0}^\infty |b_{2\ell+1}| = 1$, and
$$\forall x,y\in \R,\qquad \cos(x-y) \,=\, \sqrt{2} \sum_{\ell= 0}^\infty b_{2\ell+1} \frac{1}{2\pi} \int_{-\pi}^\pi f\big( (2\ell+1)x-t\big)g\big(t-(2\ell+1)y\big) {\rm d}t\,.$$
\end{theorem}

An explicit formula for the sequence $\{b_{2\ell+1}\}_{\ell= 0}^\infty$ can be extracted as follows from the proof of~\cite[Thm.~1]{Krivine:79a}. For any $\ell\geq 0$, define $a_{2\ell}=0$,
\begin{align*}
a_{2\ell+1} \,&=\, (-1)^\ell \cos\Big(\frac{(2\ell+1)\pi}{4}\Big)\frac{16}{\pi^2(2\ell+1)^4}\Big( \frac{1}{2\ell+1}-(-1)^\ell\frac{\pi}{4}\Big),\\
b_1 \,&=\, \frac{\sqrt{2}(\pi/4)^3}{3a_1}\,,
\end{align*}
and for $\ell>0$,
\[
b_{2\ell+1}\,=\, -\frac{1}{a_1}\sum_{\substack{d|(2\ell+1)\\d\neq 1}}
a_d b_{\frac{2\ell+1}{d}}\,.
\]
Then $|a_{2\ell+1}| = O(1/\ell^4)$, from which one deduces the crude bound $|b_{2\ell+1}| = O(1/\ell^{2})$.

\begin{figure}[ht]
\begin{protocol*}{Two-dimensional rounding procedure}
\begin{step}
\item Let $\eps>0$ and, for $j,k\in \{1,\ldots,n\}$ let $x_j,y_k\in\C$ with $|x_j| = |y_k|=1$, be given as input. Let $f,g,C$ and $\{b_{2\ell+1}\}_{\ell=0}^\infty$ be as in \expref{Theorem}{thm:krivine}.
\item\label{step:angle} For every $j,k$ let $\theta_j\in[0,2\pi)$ (resp.\ $\phi_k\in[0,2\pi)$) be the angle that $x_j$ (resp.\ $y_k$) makes with the $x$-axis.
\item\label{step:kt} Select $t \in [-\pi,\pi]$ uniformly at random. Let $L = \lceil C/\eps\rceil$ and $p = 1-\sum_{\ell=L+1}^\infty |b_{2\ell+1}|$. Select $\ell\in \{-1,0,\ldots,L \}$ with probability $\Pr(-1)=1-p$ and $\Pr(\ell) = |b_{2\ell+1}|$ for $\ell\in \{ 0,\ldots,L\}$.
\item\label{step:lm} For every $j,k$, if $\ell\geq 0$ then set $\lambda_j \eqdef \text{sign}(b_{2\ell+1})f((2\ell+1)\theta_j -t)$ and $\mu_k \eqdef g(t-(2\ell+1)\phi_k )$. Otherwise, set $\lambda_j=0$, $\mu_k=0$.
\item Return $(\lambda_j)_{j\in\{1,\ldots,n\}}$ and $(\mu_k)_{k\in\{1,\ldots,n\}}$.
\end{step}
\end{protocol*}
\caption{The two-dimensional rounding algorithm takes as input real $2$-dimensional unit vectors. It returns real numbers of absolute value at most $1$.}
\label{fig:twodimrounding}
\end{figure}

\expref{Figure}{fig:twodimrounding} describes a two-dimensional rounding scheme derived from \expref{Theorem}{thm:krivine}. The following claim states its correctness in a way that will be useful for us later.

\begin{claim}\label{claim:twodim}
Let $\eps>0$ and for every  $j,k\in \{1,\ldots,n\}$ let $x_j,y_k\in\C$  satisfy $|x_j| = |y_k|=1$. Then the rounding procedure described in \expref{Figure}{fig:twodimrounding} runs in time $\poly(n,1/\eps)$ and returns $\lambda_j,\mu_k\in \R$ with $|\lambda_j|,|\mu_k|\leq 1$ for every $j,k\in \{1,\ldots,n\}$, and
\begin{equation}\label{eq:lambda mu identity}
\E\big[\, \lambda_j\,\mu_k\,\big]\,=\,\frac{1}{\sqrt{2}}\, \Re\left(x_j\overline{y_k}\right) + \eps \langle x_j', y_k'\rangle\,,
\end{equation}
where $x_j',y_k'\in L_2(\R)$ are such that $\|x_j'\|_2,\|y_k'\|_2\leq 1$ and $\Re(\cdot)$ denotes the real part.
\end{claim}

\begin{proof}
Fix $j,k\in\{1,\ldots,n\}$ and let $\theta_j,\phi_k$ and $\lambda_j,\mu_k$ be as defined in Steps~\ref{step:angle} and~\ref{step:lm} of the rounding procedure, respectively. Applying \expref{Theorem}{thm:krivine},
\begin{align*} \E \big[ \, \lambda_j\,\mu_k \,\big] &=  \sum_{\ell=0}^L b_{2\ell+1} \frac{1}{2\pi} \int_{-\pi}^\pi f\big( (2\ell+1)\theta_j-t\big)g\big(t-(2\ell+1)\phi_k\big) {\rm d}t= \frac{1}{\sqrt{2}}\cos(\theta_j - \phi_k) + \eta_{jk}\,,
\end{align*}
where
\begin{align*}
\eta_{jk} \eqdef -\sum_{\ell= L+1}^\infty b_{2\ell+1} \frac{1}{2\pi} \int_{-\pi}^\pi f\big( (2\ell+1)\theta_j-t\big)g\big(t-(2\ell+1)\phi_k\big) {\rm d}t.
\end{align*}
By definition, $\cos(\theta_j-\phi_k) = \Re\left(x_j\overline{y_k}\right)$. Using $1-\eps \leq p \leq 1$, which follows from the bound stated in \expref{Theorem}{thm:krivine},  $\eta_{jk}$  equals $1-p\leq \eps$ times a weighted average of the product of certain values taken by $f$ and $g$, the former only depending on $\theta_j$ and the latter on $\phi_k$. Equivalently, this weighted average can be written as the inner product of two vectors $x_j'$ and $y_k'$ of norm at most $1$. Finally,  all steps of the rounding procedure can be performed in time polynomial in $n$ and $1/\eps$.
\end{proof}

\subsection{Rounding in the Hermitian case}\label{sec:hermitian}

Let $M\in M_n(M_n(\C))$ be Hermitian, and $X,Y\in \U_n(\C^d)$. For every $r\in \{1,\ldots,d\}$ define as usual $X_r,Y_r\in M_n(\C)$ by $(X_r)_{jk}=(X_{jk})_r$ and $(Y_r)_{jk}=(Y_{jk})_r$. Define $X',Y'\in M_n(\C^{2d})$ by
$$
X'_{jk}\eqdef \sum_{p=1}^d\left( \left(\frac{X_p+X_p^*}{2}\right)_{jk}e_{2p-1}+i\left(\frac{X_p-X_p^*}{2}\right)_{jk}e_{2p}\right)\in \C^{2d}
$$
and
$$
Y'_{jk}\eqdef \sum_{p=1}^d\left( \left(\frac{Y_p+Y_p^*}{2}\right)_{jk}e_{2p-1}+i\left(\frac{Y_p-Y_p^*}{2}\right)_{jk}e_{2p}\right)\in \C^{2d}\,,
$$
where $e_1,\ldots,e_{2d}$ is the canonical basis of $\C^{2d}$. 
Then $(X')(X')^* = (X')^*(X') = (XX^*+X^*X)/2 = I$, so $X'\in \U_n(\C^{2d})$ and similarly $Y'\in\U_n(\C^{2d})$. Moreover, since $M$ is Hermitian, $|M(X,Y)|=|M(X',Y')|$. This shows that for the purpose of proving the noncommutative Grothendieck inequality for Hermitian $M$ we may assume without loss of generality that the ``component matrices'' of $X,Y$ are Hermitian themselves.  Nevertheless, even in this case the rounding algorithm described in \expref{Figure}{fig:rounding} returns unitary matrices $A,B$ that are not necessarily Hermitian. The following simple argument shows how Krivine's two-dimensional rounding scheme can be applied on the eigenvalues of $A,B$ to obtain Hermitian matrices of norm $1$, at the loss of a factor $\sqrt{2}$ in the approximation. A similar argument, albeit not explicitly algorithmic, also appears in~\cite[Claim~4.7]{RegevV12a}.

\begin{figure}[t]
\begin{protocol*}{Hermitian rounding procedure}
\begin{step}
\item Let $X,Y\in M_n(\C^d)$ and $\eps>0$ be given as input.
\item\label{step:complex} Let $A, B\in M_n(\C)$ be the unitary matrices returned by the complex rounding procedure described in \expref{Figure}{fig:rounding}. If necessary, multiply $A$ by a complex phase to ensure that $M(A,B)$ is real. Write the spectral decompositions of $A,B$ as
 $$A = \sum_{j=1}^n  e^{i\theta_j} u_j u_j^*\qquad\mathrm{and}\qquad B = \sum_{k=1}^n e^{i\phi_k} v_k v_k^*\,, $$
 where $\theta_j,\phi_k\in \R$ and $u_j,v_k\in\C^n$.
\item\label{step:twodim} Apply the two-dimensional rounding algorithm from \expref{Figure}{fig:twodimrounding} to the vectors $x_j\eqdef e^{i\theta_j}$ and $y_k\eqdef e^{i\phi_k}$. Let $\lambda_j,\mu_k$ be the results.
\item Output $$A'\eqdef \sum_{j=1}^n \lambda_j u_ju_j^*\qquad\mathrm{and}\qquad B'\eqdef \sum_{k=1}^n \mu_k v_kv_k^*\,.$$
\end{step}
\smallskip
\end{protocol*}
\caption{The Hermitian rounding algorithm takes as input a pair of vector-valued matrices $X,Y\in M_n(\C^d)$. It 
outputs %
two Hermitian matrices $A',B'\in H_n$ of norm at most $1$.}
\label{fig:hermrounding}
\end{figure}

\begin{theorem}\label{thm:ncgt-herm}
Let $n$ be an integer, $M \in M_n(M_n(\C))$ Hermitian, $\eps\in (0,1)$ and $X,Y\in \U_n(\C^d)$ such that
$$\big|\,M \big( X,\, Y\big)\,\big| \,\geq\, (1-\eps)\SDP_\C(M)\,.$$
Then the rounding procedure described in \expref{Figure}{fig:hermrounding} runs in time polynomial in $n$ and $1/\eps$ and 
outputs %
a pair of Hermitian matrices $A',B'\in H_n$ with norm at most $1$ such that
$$\E\Big[\,\big|\,M \big(A',\, B'\big) \,\big|\,\Big] \,\geq\, \Big(\frac{1}{2\sqrt{2}}-\Big(1+\frac{1}{\sqrt{2}}\Big)\eps\Big)\SDP_\C(M)\,.$$
\end{theorem}

\begin{proof}
Let $A,B\in M_n(\C)$ be as defined as in Step~\ref{step:complex} of \expref{Figure}{fig:hermrounding}, and assume as in \expref{Figure}{fig:hermrounding} that $M(A,B)$ is real. By \expref{Theorem}{thm:roundingSDP} we have $\E[|M ( A, B)|] \geq \left(1/2-\eps\right)\SDP_\C(M)$. Hence to conclude it will suffice to show that for any fixed pair of matrices $A,B\in M_n(\C)$,
\begin{equation}\label{eq:desired hermitian}
\E\Big[\,\big| M(A',B')\big|\,\Big] \,\geq\, \frac{1}{\sqrt{2}} \big| M(A,B) \big| - \eps\,\SDP_\C(M)\,,
\end{equation}
where $A',B'$ are as returned by the rounding procedure and the expectation is over the random choices made in the two-dimensional rounding step, \ie, \expref{step}{step:twodim} of \expref{Figure}{fig:hermrounding}.  Applying \expref{Claim}{claim:twodim},
\begin{eqnarray}\label{eq:use krivine}
\nonumber\E \big[\,\big| M(A',B')\big|\,] &\geq& \Big| \sum_{j,k=1}^n M(u_ju_j^*,v_kv_k^*) \E \big[ \lambda_j \mu_k \,\big]\Big|\\
\nonumber&\stackrel{\eqref{eq:lambda mu identity}}{\geq}& \Big|\frac{1}{\sqrt{2}}\sum_{j,k=1}^n M(u_ju_j^*,v_kv_k^*) \Re\big( e^{i(\theta_j-\phi_k)} \big) \,\big]\Big| - \eps \,\Big| \sum_{j,k=1}^n \langle x_j', y_k'\rangle\,  M(u_ju_j^*,v_kv_k^*)\Big|\\
&=& \frac{1}{\sqrt{2}}\big|M(A,B)\big| - \eps \big|  M(W,Z)\big|\,,
\end{eqnarray}
where in the last inequality, for the first term we used that $M(u_ju_j^*,v_kv_k^*)$ is real since $M$ is Hermitian, and for the second term we defined the vector-valued matrices
\[
W\eqdef\sum_{j=1}^n \,x_j' \,  u_ju_j^* \in M_n(\C^{2n}) \qquad\text{and}\qquad Z\eqdef \sum_{k=1}^n \,y_k' \,  v_k v_k^* \in M_n(\C^{2n})\,,
\]
where multiplication of the vector $x'_j$ (resp. $y_k'$) with the scalar-valued matrix $u_ju_j^*$ (resp. $v_k v_k^*$) is taken entrywise.
Here we assumed that the vectors $x'_j, y'_k$ from \expref{Claim}{claim:twodim} lie in $\C^{2n}$, which is without loss of generality since there are $2n$ of them.
One checks that
\[
WW^*=\sum_{j=1}^n \|x_j'\|_2^2 u_ju_j^*\,.
\]
Since $\|x_j'\|_2\le 1$ for all $j\in \{1,\ldots,n\}$, it follows that
$\|WW^*\|\le 1$. Similarly
\[
\max\{\|W^*W\|,\|ZZ^*\|,\|Z^*Z\|\}\le 1\,.
\]
Applying \expref{Lemma}{lem:equiv-sdp} we obtain $R,S\in \U_n(\C^{2n(n+1)})$ such that $M(W,Z)=M(R,S)$, hence $| M(W,Z)|\leq \SDP_\C(M)$. Equation~\eqref{eq:use krivine} then implies the desired estimate~\eqref{eq:desired hermitian}.
\end{proof}

\subsection{Proof of \expref{Theorem}{thm:hermitian real equivalence}}\label{sec:real-herm-equiv}

In this section we prove \expref{Theorem}{thm:hermitian real equivalence}. We first record for future use the following simple lemma, which is an algorithmic version of~\eqref{eq:pass to convex hull R}.

\begin{lemma}\label{eq:round to orthogonal}
There exists a 
polynomial-time 
algorithm that takes as input a $4$-tensor $M\in M_n(M_n(\R))$ and two matrices $A,B\in M_n(\R)$ with $\max\{\|A\|,\|B\|\}\le 1$ and 
outputs %
two orthogonal matrices $U,V\in \O_n$ such that
$$
\sum_{i,j,k,l=1}^n M_{ijkl} A_{ij}B_{kl}\le \sum_{i,j,k,l=1}^n M_{ijkl} U_{ij}V_{kl}\,.
$$
\end{lemma}
\begin{proof}
Write the singular value decompositions of $A,B$ as $A=\sum_{i=1}^n \sigma_i e_i f_i^*$
and $B=\sum_{i=1}^n \tau_i g_i h_i^*$, where each of the sequences $(e_i)_{i=1}^n$, $(f_i)_{i=1}^n$, $(g_i)_{i=1}^n,(h_i)_{i=1}^n\subseteq \R^n$ is orthonormal, and
$\sigma,\tau\in [0,1]^n$ (since $\max\{\|A\|,\|B\|\}\le 1$). Now, $M(A,B)$ is given by
\begin{equation}\label{eq:plug singular}
\sum_{i,j=1}^n \sigma_i \tau_j M(e_i f_i^*, g_j h_j^*)\,.
\end{equation}
Fixing all $\sigma_i$, $\tau_j$ but, say, $\sigma_1$, \eqref{eq:plug singular} is a linear function of $\sigma_1$, and thus we can shift $\sigma_1$ to
either $-1$ or $1$, without decreasing~\eqref{eq:plug singular}. Proceeding in this way with the other
variables $\sigma_2,\ldots,\sigma_n,\tau_1,\ldots,\tau_n$, each one in its turn, we obtain $\e,\d\in \{-1,1\}^n$ such that if we define $U,V\in \O_n$ by $U= \sum_{i=1}^n \e_i (e_i f_i^*)$ and $V=\sum_{i=1}^n \d_i (g_i h_i^*)$ then $M(A,B)\le M(U,V)$, as required.
\end{proof}

\begin{proof}[Proof of \expref{Theorem}{thm:hermitian real equivalence}]
We first prove the implication~\ref{part:hermite-1}~$\implies$~\ref{part:hermite-2}.

For any $A\in M_{2n}(\C)$ define $A_1,A_2,A_3,A_4 \in M_n(\C)$ through
the block decomposition 
\[
A = \left(\begin{array}{cc} A_{1} & A_{2}\\A_{3} &
    A_{4}\end{array}\right)\,.
\]
Given $M\in M_n(M_n(\R))$ let $M'\in M_{2n}(M_{2n}(\R))$ be such that $M'(A,B)=M(\Re A_2,\Re B_2)$ for every $A,B\in H_{2n}$. Formally, for every $i,j,k,l\in \{1,\ldots,2n\}$ we have
$$
M'_{ijkl}=\frac14\left\{\begin{array}{ll}
M_{i,j-n,k,l-n} & \mathrm{if}\ (i,k)\in \{1,\ldots,n\}^2\ \mathrm{and}\ (j,l)\in \{n+1,\ldots,2n\}^2, \\
 M_{j,i-n,l,k-n} &\mathrm{if}\ (j,l)\in \{1,\ldots,n\}^2   \ \mathrm{and}\ (i,k)\in \{n+1,\ldots,2n\}^2, \\
 M_{i,j-n,l,k-n} &\mathrm{if}\ (i,l)\in \{1,\ldots,n\}^2   \ \mathrm{and}\ (j,k)\in \{n+1,\ldots,2n\}^2, \\
 M_{j,i-n,k,l-n} &\mathrm{if}\ (j,k)\in \{1,\ldots,n\}^2   \ \mathrm{and}\ (i,l)\in \{n+1,\ldots,2n\}^2,\\
 0 &\mathrm{otherwise.}
\end{array}\right.$$
Then $M'$ is Hermitian. Apply the algorithm $\alg^*$ (whose existence is the premise of~\ref{part:hermite-1} of \expref{Theorem}{thm:hermitian real equivalence}) to $M'$ to get $A,B\in H_{2n}$ such that $\opt_\C^*(M)\le K|M(A,B)|$.  Since $A,B\in H_{2n}$ we have $A_3=A_2^*$ and $B_3=B_2^*$. Because $\max\{\|A\|,\|B\|\}\le 1$ also $\max\{\|\Re A\|,\|\Re B\|\}\le 1$. By \expref{Lemma}{eq:round to orthogonal} we can therefore efficiently find $U,V\in \O_n$ such that
\begin{align*}
KM(U,V) &\ge K|M(\Re A,\Re B)|=K|M'(A,B)|\ge \opt_\C^*(M')\\[1ex]
&=\sup_{\substack{C,D\in H_{2n}\\ \|C\|,\|D\|\le 1}}\left|M'(C,D)\right|
\ge \sup_{S,T\in \O_n} \left|M'\left(\left(\begin{array}{cc} 0 & S\\
        S^* & 0\end{array}\right),\left(\begin{array}{cc} 0 & T\\T^* &
        0\end{array}\right)\right)\right|\\
&=\sup_{S,T\in O_n} M(S,T)=\opt_\R(M)\,.
\end{align*}

To prove the reverse implication \ref{part:hermite-2}~$\implies$~\ref{part:hermite-1}, define $\psi:M_{2n}(\R)\to H_n$ by
\begin{equation}\label{eq:def psi}
\psi\left(\left(\begin{array}{cc} A_{1} & A_{2}\\A_{3} & A_{4}\end{array}\right)\right)\eqdef \frac14\left(A_1+A_1^*+A_4+A_4^*\right)+\frac{i}{4}\left(A_2-A_2^*+A_3^*-A_3\right)\in H_n\,.
\end{equation}
Suppose that $M\in M_n(M_n(\C))$ is Hermitian and define $M''\in M_{2n}(M_{2n}(\R))$ by requiring that $
M''(A,B)= M\left(\psi(A),\psi(B)\right)
$ for every $A,B\in M_{2n}(\R)$.  Note that $M(\psi(A),\psi(B))\in \R$
since $M, \psi(A),\psi(B)$ are all Hermitian. Apply the algorithm
$\alg$ of part $2)$ to $M''$,  obtaining two orthogonal matrices
$U,V\in \O_{2n}$ satisfying 
\[
\opt_\R(M'')\le KM''(U,V)=KM(\psi(U),\psi(V))\,.
\]
By \expref{Lemma}{lem:psi bound} below we have $\max\{\|\psi(U)\|,\|\psi(V)\|\}\le 1$. Moreover, using \expref{Lemma}{lem:hermitian norm identity} below we have
\begin{equation}
  \begin{split}
    \opt_\R(M'') &=\sup_{\substack{A,B\in M_{2n}(\R)\\\|A\|,\|B\|\le
        1}}\big|M(\psi(A),\psi(B))\big|\\
    &\ge \sup_{\substack{X,Y\in H_n\\\|X\|,\|Y\|\le
        1}}\left|M\left(\psi\left(\left(\begin{array}{cc} \Re X & \Im
              X\\-\Im X & \Re
              X\end{array}\right)\right),\psi\left(\left(\begin{array}{cc}
              \Re Y & \Im Y\\-\Im Y& \Im
              Y\end{array}\right)\right)\right)\right|\,,\label{eq:M'' opt}
  \end{split}
\end{equation}
where $\Re(\cdot)$ and $\Im(\cdot)$ denote the real and imaginary parts,
respectively. 
Observe that for every $X\in H_n$ we have
\[
\psi\left(\left(\begin{array}{cc} \Re X & \Im X\\-\Im X & \Re
      X\end{array}\right)\right)=X\,.
\]
Consequently the rightmost term in~\eqref{eq:M'' opt} equals $\opt_\C^*(M)$. Therefore  $\opt_\C^*(M)\le K|M(\psi(U),\psi(V))|$, so that the algorithm that
outputs %
$\psi(U),\psi(V)$ has the desired approximation factor.
\end{proof}

\begin{lemma}\label{lem:psi bound}
Let $\psi: M_{2n}(\R)\to H_n$ be given as in~\eqref{eq:def psi}. Then $\|\psi(Y)\|\le \|Y\|$ for all $Y\in M_{2n}(\R)$.
\end{lemma}

\begin{proof}
For $Y\in M_{2n}(\R)$ write $Z=(Y+Y^*)/2$. Setting
\[
Z = \left(\begin{array}{cc} Z_{1} & Z_{2}\\Z_{3} &
    Z_{4}\end{array}\right)\,,
\]
where $Z_1,Z_2,Z_3,Z_4\in M_n(\R)$, we then have $
\psi(Y)=(Z_{1}+Z_{4})/2+i(Z_{2}-Z_{3})/2$. Take $z\in \C^n$ and write $z=x+iy$, where $x,y\in \R^n$. Then
\begin{align*}
\Re\big( \langle\psi(Y) z,z\rangle \big) &= \frac{\langle Z_{1} x,x\rangle  + \langle Z_{4} y,y\rangle - \langle Z_{2} y,x\rangle - \langle Z_{3} x,y\rangle}{2} + \frac{\langle Z_{1} y,y\rangle  + \langle Z_{4} x,x\rangle + \langle Z_{2} x,y\rangle + \langle Z_{3} y,x\rangle}{2}\\[1ex]
&= \frac{1}{2} \left\langle Z \begin{pmatrix} x \\ -y\end{pmatrix},\begin{pmatrix} x \\ -y\end{pmatrix}\right\rangle + \frac{1}{2} \left\langle Z \begin{pmatrix} y \\ x\end{pmatrix},\begin{pmatrix} y \\ x\end{pmatrix}\right\rangle\,.
\end{align*}
Since $\psi(Y)$ is Hermitian, it follows that $\|\psi(Y)\|\le \|Z\|\le \|Y\|$.
\end{proof}

\begin{lemma}\label{lem:hermitian norm identity}
For every $X\in H_n$ we have $\displaystyle \left\|\left(\begin{array}{cc} \Re X & \Im X\\-\Im X & \Re X\end{array}\right)\right\|= \|X\|$.
\end{lemma}

\begin{proof}
Write
\[
Z=\left(\begin{array}{cc} \Re X & \Im X\\-\Im X & \Re
    X\end{array}\right)\in M_{2n}(\R)\,.
\]
Since $X$ is Hermitian, $Z$ is symmetric. For every $a,b\in \R^n$,
\begin{eqnarray*}
\left\langle Z \begin{pmatrix} a \\ b\end{pmatrix},\begin{pmatrix} a \\ b\end{pmatrix}\right\rangle &=&\left\langle (\Re X)a,a\right\rangle + \left\langle (\Im X)b,a\right\rangle - \left\langle (\Im X)a,b
\right\rangle + \left\langle (\Re X)b,b
\right\rangle\\&=&
\Re\left(\langle X(a-ib),a-ib\rangle \right).
\end{eqnarray*}
Since $Z$ is symmetric and $X$ is Hermitian, it follows that $\|Z\|=\|X\|$.
\end{proof}

\section{Direct rounding in the real case}\label{sec:real}

We now describe and analyze a different rounding procedure for the real case of the noncommutative Grothendieck inequality. The argument is based on the proof of the noncommutative Grothendieck inequality due to Kaijser~\cite{Kaijser83}, which itself has a structure similar to the proof of the classical Grothendieck inequality due to Krivine~\cite{Kri73} (see also~\cite{Jam85}, the ``Notes and Remarks'' section of Chapter~1 of~\cite{DJT95}, and the survey article~\cite{JL01}), and uses ideas from~\cite{Pisier78NCGT} to extend that proof to the non-commutative setting.

\begin{figure}[t]
\begin{protocol*}{Real rounding procedure}
\begin{step}
\item Let $X,Y\in M_n(\R^d)$ be given as input, and let $\e\in \{-1,1\}^d$ be chosen uniformly at random.
\item \label{step:real-svd} Set $\Ye \eqdef \langle \e,Y\rangle$. Write the singular value decomposition of $\Ye$ as $\Ye = \sum_{i=1}^n \,t_i(\e)\, u_i(\e) v_i(\e)^*$, where $t_i(\e)\in [0,\infty)$ and $u_i(\e),v_i(\e)\in \R^n$. Define $$(\Ye)_\tau \eqdef \sum_{i=1}^n \min\{t_i(\e),\tau\} u_i(\e) v_i(\e)^*\,,$$
     where $\tau\eqdef \sqrt{3}/2$.
\item Let $X(\e)\in M_n(\R)$ of norm at most $1$ be such that $$M(X(\e),(\Ye)_\tau) = \max_{\substack{X\in M_n(\R)\\\|X\|\leq 1}} |M(X,(\Ye)_\tau)|\,.$$
\item Output the pair $A=X(\e)$ and $B= \frac{1}{\tau}(\Ye)_\tau$.
\end{step}
\end{protocol*}
\caption{The real rounding algorithm takes as input $X,Y\in M_n(\R^d)$. It 
outputs %
two real matrices $A,B\in M_n(\R)$ of norm at most $1$.}
\label{fig:realrounding}
\end{figure}

\begin{theorem}\label{thm:ncgt-real}
Given $n\in \N$, $M \in M_n(M_n(\R))$ and $\eta\in (0,1/2)$, suppose that $X,Y\in \O_n(\R^d)$ are such that
\begin{equation}\label{eq:eta asusmption}
\big|\,M \big( X,\, Y\big)\,\big| \,\geq\, (1-\eta)\SDP_\R(M)\,,
\end{equation}
where $\SDP_\R(M)$ is defined in~\eqref{eq:def sdpR}.
Then the rounding procedure described in \expref{Figure}{fig:realrounding} runs in time polynomial in $n$ and 
outputs %
a pair of real matrices $A,B\in M_n(\R)$ with norm at most $1$ such that
\begin{equation}\label{eq:real desired}
\E\Big[\,\big|\,M \big(A,\, B\big) \,\big|\,\Big] \,\geq\, \frac{(1-2\eta)^2}{3\sqrt{3}}\SDP_\R(M)\,.
\end{equation}
\end{theorem}
Note that by \expref{Lemma}{eq:round to orthogonal} we can also efficiently convert the matrices $A,B$ to orthogonal matrices $U,V\in \O_n$ without changing the approximation
guarantee.

The proof of \expref{Theorem}{thm:ncgt-real} relies on two claims that
are used to bound the error that results from the truncation step
(\expref{step}{step:real-svd}) in the rounding procedure in \expref{Figure}{fig:realrounding}. The first claim plays the same role as \expref{Claim}{claim:haag-id} did in the complex case.

\begin{claim}\label{claim:moment-bound}
Fix $X\in M_n(\R^d)$ and let $\e\in \{-1,1\}^d$ be chosen uniformly at random. Set $\Xe= \langle \eps,X\rangle$. Then
\begin{align*}
\E_\e \big[(\Xe\Xe^*)^2\big] &\leq (XX^*)^2 + 2 \|X^*X\| XX^*\,,\\
\E_\e \big[(\Xe^* \Xe)^2\big] &\leq (X^*X)^2 + 2 \|XX^*\| X^*X\,.
\end{align*}
\end{claim}

\begin{proof}
Define the symmetric real vector-valued matrix $Z$ by
\[
Z \eqdef \begin{pmatrix} 0 & X \\ X^* & 0 \end{pmatrix}.
\]
Following the proof of Lemma~1.1 in~\cite{Pisier78NCGT} (which establishes a bound analogous to the one proved in \expref{Claim}{claim:haag} for the case of i.\,i.\,d.\ $\{\pm 1\}$ random variables) and defining as usual $Z_r\in M_n(\R)$ by $(Z_r)_{ij} = (Z_{ij})_{r}$ for every $r\in \{1,\ldots,d\}$, we have
\begin{equation}
\begin{split}
\label{eq:expand Zs}
\E_\e \big[\langle \e, Z\rangle^4\big] &= \sum_{r=1}^d Z_r^4 + \sum_{\substack{r,s\in \{1,\ldots,d\}\\ r\neq s}}\,\big( Z_r^2 Z_s^2 + Z_rZ_sZ_rZ_s + Z_r Z_s^2 Z_r\big) \\
&= \Big( \sum_{r=1}^d Z_r^2\Big)^2 + \sum_{\substack{r,s\in
    \{1,\ldots,d\}\\r< s}}  \,\big(Z_r Z_s + Z_s Z_r\big)^2.
\end{split}
\end{equation}
Using the inequality $(A+B)(A+B)^* \leq 2(AA^*+BB^*)$, which follows from $(A-B)(A-B)^*\geq 0$ for all $A,B\in M_n(\R)$, we can bound the second sum in~\eqref{eq:expand Zs} as follows.
\begin{align*}
\sum_{\substack{r,s\in \{1,\ldots,d\}\\r< s}}  \,\big(Z_r Z_s + Z_s Z_r\big)^2 &\leq 2\,\Big(\sum_{r=1}^d \,Z_r \Big(\sum_{\substack{s\in \{1,\ldots,d\}\\s>r}} Z_s^2 \Big) Z_r + \sum_{s=1}^d \,Z_s \Big(\sum_{\substack{r\in \{1,\ldots,d\}\\r< s}}  Z_r^2 \Big) Z_s\Big)\\
&= 2\sum_{r=1}^d \,Z_r\Big(\sum_{\substack{s\in \{1,\ldots,d\}\\ s\neq r}} Z_s^2 \Big) Z_r\\
& \leq 2\sum_{r=1}^d Z_r\Big(\sum_{s=1}^d Z_s^2 \Big) Z_r\,.
\end{align*}
Replacing $Z_r$ by its definition and using $ABA^* \leq \|B\| AA^*$, which holds true for every positive semidefinite $B\in M_n(\R)$, we
arrive at the following matrix inequality:
\begin{align*}
\E_\e \big[\langle \e, Z\rangle^4\big] &= \begin{pmatrix} \E_\e \big[ (\Xe \Xe^\dagger)^2 \big] & 0 \\ 0 & \E_\e \big[(\Xe^\dagger \Xe)^2\big] \end{pmatrix}\\[1ex] &\leq \begin{pmatrix} (XX^\dagger)^2 & 0 \\ 0 & (X^\dagger X)^2 \end{pmatrix} + 2\,\begin{pmatrix} \|X^\dagger X\| \,XX^\dagger & 0 \\ 0 & \| X X^\dagger \| \,X^\dagger X\end{pmatrix}.\end{align*}
The inequality above implies a separate matrix inequality for both diagonal blocks, proving the claim.
\end{proof}

The second claim appears as Lemma~2.3 in~\cite{Kaijser83} (in the complex case). We include a short proof for the sake of completeness.

\begin{claim}\label{claim:truncate}
Let $Y\in M_n(\R)$. For any $\tau>0$, there exists a decomposition $Y = Y_{\tau} + Y^\tau$ such that $\|Y_\tau\|_\infty \leq \tau$ and
$$Y^\tau (Y^\tau)^* \le \frac{1}{(4\tau)^2} (YY^*)^2\qquad\text{and}\qquad (Y^\tau)^*Y^\tau  \le \frac{1}{(4\tau)^2} (Y^*Y)^2\,.$$
\end{claim}

\begin{proof}
Define $Y_\tau$ by ``truncating'' the singular values of $Y$ at $\tau$, as is done in step 2 of the rounding procedure described in \expref{Figure}{fig:realrounding}, so that $\|Y_\tau\|\leq\tau$. Define $Y^\tau=Y-Y_\tau$. By definition, $Y$, $Y_\tau$ and $Y^\tau$ all have the same singular vectors, and the non-zero singular values of $Y^\tau$ are of the form $\mu-\tau$, where $\mu\geq \tau$ is a singular value of $Y$. Using the bound $|\mu-\tau|\leq \mu^2/(4\tau)$ valid for any $\mu\geq \tau$, any nonzero eigenvalue $\lambda = (\mu-\tau)^2$ of $Y^\tau (Y^\tau)^*$ (resp.\ $(Y^\tau)^*Y^\tau $) satisfies $\lambda \leq \mu^4/(4\tau)^2$, which proves the claim.
\end{proof}

\begin{proof}[Proof of \expref{Theorem}{thm:ncgt-real}]
We shall continue using the notation introduced in \expref{Figure}{fig:realrounding}, \expref{Claim}{claim:moment-bound} and \expref{Claim}{claim:truncate}. Let $X,Y\in\O_n(\R^d)$ satisfying~\eqref{eq:eta asusmption}. For every $\e\in \{-1,1\}^d$ and any $\tau>0$ let  $$\Ye = \langle \eps,Y\rangle = (\Ye)_{\tau} + \Ye^\tau$$ be the decomposition promised by \expref{Claim}{claim:truncate}. Combining the bound from \expref{Claim}{claim:moment-bound} with the one from \expref{Claim}{claim:truncate}, we see that
\begin{equation}\label{eq:ytau-norm}
\big\|\E_\e \big[\Ye^\tau(\Ye^\tau)^*\big]\big\| \le \frac{1}{(4\tau)^2}\big(\|YY^*\|^2 + 2\|Y^*Y\|\|YY^*\|\big) = \frac{3}{(4\tau)^2}\,,
\end{equation}
where the final step of~\eqref{eq:ytau-norm} follows from $Y\in\O_n(\R^d)$, and the same bound holds on $\|\E_\e \big[(\Ye^\tau)^*\Ye^\tau\big]\|$.
We have
\begin{equation}
\begin{split}
\big| M \big(X,\, Y\big) \big| &= \big| \E_\e\big[M (\Xe,\,\Ye)\big] \big|
\leq \big|\E_\e\big[ M (\Xe,\, (\Ye)_\tau)\big]\big| + \big|\E_\e\big[ M(\Xe,\, \Ye^\tau)\big]\big|\\
&\leq \big|\E_\e \big[M (\Xe,\, (\Ye)_\tau)\big]\big| +
\frac{\sqrt{3}}{4 \tau} \SDP_\R(M)\,,\label{eq:ncgt-real-0}
\end{split}
\end{equation}
where the last inequality in~\eqref{eq:ncgt-real-0} follows from the definition of $\SDP_\R(M)$ and~\eqref{eq:ytau-norm}.

To bound the first term in~\eqref{eq:ncgt-real-0}, let
\[
\Xe = \sum_{i=1}^n s_i(\e) w_i(\e) z_i(\e)^*\qquad\text{and}\qquad
(\Ye)_\tau = \sum_{i=1}^n t_i'(\e) u_i(\e) v_i(\e)^*\,,
\]
where $s_i(\e),t_i'(\e)\ge 0$ and $u_i(\e),v_i(\e),w_i(\e),z_i(\e)\in\R^n$, be the singular value decompositions of $\Xe,(\Ye)_\tau$, respectively. Set $M_{i,j}(\e)\eqdef  M(w_i(\e) z_i(\e)^*, u_j(\e) v_j(\e)^*)$. With these definitions,
\begin{align}
\big|\E_\e\left[ M (\Xe,\, (\Ye)_\tau)\right]\big| &= \Big|\E_\e \Big[\sum_{i,j=1}^n\, M_{i,j}(\e) s_i(\e) t_j'(\e)\Big] \Big|
 \leq \E_\e\Big[\sum_{i=1}^n \,s_i(\e)\, \Big| \sum_{j=1}^n M_{i,j}(\e) t_j'(\e)\Big|\Big] \nonumber\\
&\leq \Big(\E_\e\Big[\sum_{i=1}^n \,s_i(\e)^2 \Big| \sum_{j=1}^n M_{i,j}(\e) t_j'(\e)\Big|\Big]\Big)^{1/2} \Big(\E_\e \Big[\sum_{i=1}^n \Big| \sum_{j=1}^n M_{i,j}(\e) t_j'(\e) \Big|\Big]\Big)^{1/2}.\label{eq:ncgt-real-1}
\end{align}
Note that the rightmost term in~\eqref{eq:ncgt-real-1} is at most
\[
\Big(\tau \E_\e\Big[M(X(\e),B(\e))\Big]\Big)^{1/2}\,,
\]
where $B(\e) = ({1}/{\tau})(\Ye)_\tau$ and $X(\e)$ is as defined in
the rounding procedure in \expref{Figure}{fig:realrounding}. To bound
the leftmost term in~\eqref{eq:ncgt-real-1}, define
\[
W_\e \eqdef \sum_{i=1}^n (r_i(\e) s_i(\e)^2) w_i(\e) z_i(\e)^*\,,
\]
where $r_i(\e)$ is the sign of $\sum_{j=1}^n M_{i,j}(\e) t_j'(\e)$, so that $$\sum_{i=1}^n \,s_i(\e)^2 \left| \sum_{j=1}^n M_{i,j}(\e) t_j'(\e)\right| = M(W_\e,(\Ye)_\tau)\,.$$ Moreover, by definition $W_\e W_\e^* = (\Xe\Xe^*)^2$, so using \expref{Claim}{claim:moment-bound} we have
\begin{align*}
 \big\|\E_\e\big[W_\e^*W_\e\big]\big\|&=\big\|\E_\e \left[(\Xe^*\Xe)^2\right]\big\| \leq \| X^*X \|^2+ 2\| XX^* \|\| X^*X \| = 3\,,
\end{align*}
and the same bound holds for $\|\E_\e \big[W_\e W_\e^*\big]\|$. By the definition of $(Y_\e)_\tau$ we also have
$$\max\Big\{\big\|\E_\e\left[(\Ye)_\tau^* (\Ye)_\tau\right]\big\|,\, \big\|\E_\e\left[ (\Ye)_\tau (\Ye)_\tau^*\right]\big\| \Big\}\leq \tau^2\,.$$
 Hence
$$\E_\e\Big[ \sum_{i=1}^n \,s_i(\e)^2 \Big| \sum_{j=1}^n M_{i,j}(\e) t_j'(\e)\Big|\Big] \,=\, \sqrt{3}\tau \,\E_\e\Big[  M\Big(\frac{W_\e}{\sqrt{3}},\,\frac{(\Ye)_\tau}{\tau}\Big)\Big]\,  \leq\, \sqrt{3}\tau\,\SDP_\R(M)\,,$$
where we used the definition of $\SDP_\R(M)$. Finally, combining~\eqref{eq:ncgt-real-0} and~\eqref{eq:ncgt-real-1} with the bounds shown above we obtain
\begin{equation}\label{eq:just before tau choice}
(1-\eta)\SDP_\R(M)\stackrel{\eqref{eq:eta asusmption}}{\le}\big| M \big(X ,\, Y\big)\big| \,\leq\, \sqrt[4]{3}\tau \sqrt{\SDP_\R(M) \cdot \E_\e\left[M(X(\e),B(\e))\right]} + \frac{\sqrt{3}}{4\tau}\SDP_\R(M)\,.
\end{equation}
Setting $\tau = \sqrt{3}/2$ in~\eqref{eq:just before tau choice} and simplifying leads to the desired bound~\eqref{eq:real desired}.

All steps in the rounding procedure can be performed efficiently. The calculation of $X(\e)$ in the third step of \expref{Figure}{fig:realrounding} can be expressed as a semidefinite program and solved in time polynomial in $n$. Alternatively, one may directly compute $X(\e)$ as follows. Write the polar decomposition 
\[
\Tr_2( M^* (I \otimes (\Ye)_\tau)) = Q P \in M_n(\R)\,,
\]
where the partial trace is taken with respect to the second tensor, $Q$ is an orthogonal matrix and $P$ is positive semidefinite. The optimal choice of $X(\e)$ corresponds to the real matrix of norm at most $1$ maximizing 
\[
|\Tr(X(\e)\cdot QP)|\,=\,\big|\Tr\big(X(\e)\cdot\Tr_2( M^* (I \otimes (\Ye)_\tau)) \big) \big|\,=\, |M(X(\e),(\Ye)_\tau)|\,,
\]
and this is achieved by taking $X(\e) = Q^*$.
\end{proof}

\section{Some applications}\label{sec:applications}

Before presenting the details of our applications of \expref{Theorem}{thm:alg}, we observe that the problem of computing $\mathrm{Opt}_\R(M)$ is a
rather versatile optimization problem, perhaps more so than what one might initially
guess from its definition. The main observation is that by considering matrices $M$
which only act non-trivially on certain diagonal blocks of the two variables $U,V$ that
appear in the definition of $\opt_\R(M)$, these variables can each be thought of as a sequence of multiple matrix
variables, possibly of different shapes but all with operator norm at most $1$. This
allows for some flexibility in adapting the noncommutative Grothendieck optimization problem
to concrete settings, and we explain the transformation in detail next.

For every $n,m \ge 1$, let $M_{m,n}(\R)$ be the vector space of real $m\times n$ matrices. Given integers $k,\ell\ge 1$ and sequences of integers $(m_i),(n_i)\in \N^k$, $(p_j),(q_j)\in \N^\ell$, we define $\Bil(k,\ell;(m_i),(n_i),(p_j),(q_j))$, or simply $\Bil(k,\ell)$ when the remaining sequences are clear from context, as the
set of all
$$f:\Big(\bigoplus_{i=1}^k M_{m_i,n_i}(\R)\Big)\times\Big( \bigoplus_{j=1}^\ell M_{p_j,q_j}(\R) \Big)\,\to\, \R$$
that are linear in both arguments. Concretely, $f\in\Bil(k,\ell)$ if and only if there exists real coefficients $\alpha_{irs,juv}$ such that for every $(A_i)\in \bigoplus_{i=1}^k M_{m_i,n_i}(\R)$ and $(B_j)\in \bigoplus_{j=1}^\ell M_{p_j,q_j}(\R)$,
\begin{equation}\label{eq:bilf-def}
f\big( (A_i)_{i\in \{1,\ldots,k\}},(B_j)_{j\in\{1,\ldots,\ell\}}\big) = \sum_{i=1}^k\sum_{j=1}^{\ell} \sum_{r=1}^{m_i}\sum_{s=1}^{n_i} \sum_{u=1}^{p_j}\sum_{v=1}^{q_j}\,\alpha_{irs,juv} \,(A_i)_{rs} (B_j)_{uv}\,.
\end{equation}
For integers $m,n\geq 1$, let $\O_{m,n} \subset M_{m,n}(\R)$ denote the
set of all $m \times n$ real matrices $U$ such that $UU^* = I$ if $m\leq n$ and $U^*U = I$ if $m\ge n$.
If $m=n$ then $\O_{n,n}=\O_n$
is the set of orthogonal matrices; $\O_{n,1}$ is the set of all $n$-dimensional
unit vectors; $\O_{1,1}$ is simply the set $\{-1,1\}$.
Given $f\in \Bil(k,\ell)$, consider the quantity
\[
 \mathrm{Opt}_\R(f)\eqdef \sup_{\substack{(U_i)\in \bigoplus_{i=1}^k\O_{m_i,n_i}\\ (V_j)\in \bigoplus_{j=1}^\ell\O_{p_j,q_j}}} f\big((U_i),(V_j)\big)\,.
\]
Note that this definition coincides with the definition of $\mathrm{Opt}_\R(f)$ given in the introduction whenever $f\in \Bil(1,1;n,n,n,n)$. The proof of the following proposition shows that the new optimization problem still belongs the framework of the noncommutative Grothendieck problem.

\begin{proposition}\label{prop:usefulform}
There exists a 
polynomial-time 
algorithm that takes as input $k,\ell\in\N$, $(m_i),(n_i)\in\N^k$, $(p_j),(q_j)\in\N^\ell$ and $f\in \Bil(k,\ell;(m_i),(n_i),(p_j),(q_j))$ and 
outputs %
$(U_i) \in \bigoplus_{i=1}^k \O_{m_i,n_i}$ and $(V_j)\in \bigoplus_{j=1}^\ell\O_{p_j,q_j}$ such that
$$ \mathrm{Opt}_\R(f)\,\le\, O(1)\cdot f\big( (U_i),(V_j)\big)\,.$$
Moreover, the implied constant in the $O(1)$ term can be taken to be any number
larger than $2\sqrt{2}$.
\end{proposition}

\begin{proof}
Let $k,\ell\in\N$, $(m_i),(n_i)\in\N^k$, $(p_j),(q_j)\in\N^\ell$ be given, and define $$ m\eqdef\sum_{i=1}^k m_i,\qquad n\eqdef\sum_{i=1}^k n_i,\qquad p\eqdef\sum_{j=1}^\ell p_j,\qquad q\eqdef\sum_{j=1}^\ell q_j\,,$$ and $t\eqdef \max\{m,n,p,q\}$. We first describe how   $\bigoplus_{i=1}^k M_{m_i,n_i}(\R)$ 
(and $\bigoplus_{j=1}^\ell M_{p_j,q_j}(\R)$, 
respectively) %
can be identified with a subset of $M_t(\R)$ consisting of 
block-diagonal matrices.
For any $i\in \{1,\ldots,k\}$ and $r\in \{1,\ldots,  m_i\}$, $s\in \{1,\ldots,n_i\}$, let $F_{r,s}^i \in M_t(\R)$ be the matrix that has all entries equal to $0$ except the entry in position $(r+\sum_{j<i} m_j ,s+\sum_{j<i} n_j)$, which equals $1$. Similarly, for any $j\in \{1,\ldots,\ell\}$ and $u\in \{1,\ldots, p_j\}$, $v\in \{1,\ldots, q_j\}$ we let $G_{u,v}^j \in M_t(\R)$ be the matrix that has all entries equal to $0$ except the entry in position $(u+\sum_{i<j} p_i ,v+\sum_{i<j} q_i )$, which equals $1$. Define linear maps $\Phi:\bigoplus_{i=1}^k M_{m_i,n_i}(\R)\to M_t(\R)$ and $\Psi: \bigoplus_{j=1}^\ell M_{p_j,q_j}(\R) \to M_t(\R)$ by
$$ \Phi\big( (A_i) \big) \eqdef \sum_{i=1}^k \sum_{r=1}^{m_i}\sum_{s=1}^{n_i}\, (A_i)_{r,s}\, F^i_{r,s} \qquad\text{and}\qquad  \Psi\big( (B_j) \big) \eqdef \sum_{j=1}^\ell \sum_{u=1}^{p_j} \sum_{v=1}^{q_j}\, (B_j)_{u,v}\, G^j_{u,v}\,.$$
{}From the definition, one verifies that
\begin{equation}\label{eq:useful-1}
\|\Phi((A_i))\| = \max_{i\in \{1,\ldots,k\}} \|A_i\|\qquad\text{and}\qquad \|\Psi((B_j))\| = \max_{j\in \{1,\ldots,\ell\}} \|B_j\|\,.
\end{equation}
Let $f\in\Bil(k,\ell)$ and
let
 $(\alpha_{irs,juv})$
be
real coefficients as in~\eqref{eq:bilf-def}. Define $M\in M_t(M_t(\R))$ by $$M(F_{r,s}^i,G_{u,v}^j) \eqdef \alpha_{irs,juv}\,,$$ and $M(A,B)\eqdef 0$ if $A$ or $B$ are orthogonal to the linear span of the $F_{r,s}^i$ or $G_{u,v}^j$,  %
respectively. (Recall that the notation $M(A,B)$ was introduced at the beginning of \expref{Section}{sec:proof}.) With this definition,
for any $(A_i) \in \bigoplus_{i=1}^k M_{m_i,n_i}(\R)$ and $(B_j)\in \bigoplus_{j=1}^\ell M_{p_j,q_j}(\R)$ we have
$$M\big(\Phi((A_i)),\Psi((B_j))\big) \,=\, f\big((A_i),(B_j)\big)\,.$$

We claim that $\mathrm{Opt}_\R(f) = \mathrm{Opt}_\R(M)$, where $\mathrm{Opt}_\R(M)$ is as defined in the introduction. Indeed, by~\eqref{eq:useful-1} for any $(U_i) \in \bigoplus_{i=1}^k \O_{m_i,n_i}$ and $(V_j)\in \bigoplus_{j=1}^\ell \O_{p_j,q_j}$ we have $\|\Phi((U_i))\|,\|\Psi((V_j))\|\leq 1$, hence $\mathrm{Opt}_\R(f) \le \mathrm{Opt}_\R(M)$.
Conversely, let $U,V \in \O_t$ be arbitrary, and $U',V'$ their orthogonal projections on $\mathrm{Im}(\Phi)$ and %
$\mathrm{Im}(\Psi)$, respectively. %
Then  $M(U',V')=M(U,V)$. Moreover, if $$(U'_i)\in \bigoplus_{i=1}^k M_{m_i,n_i}(\R)\qquad\mathrm{and}\qquad (V'_j)\in  \bigoplus_{j=1}^\ell M_{p_j,q_j}(\R)$$ are such that $\Phi((U'_i))=U'$ and $\Psi((V'_j))=V'$ then by~\eqref{eq:useful-1} $\max_{i\in \{1,\ldots,k\}} \|U'_i\| = \|U'\|\leq \|U\|=1$, and similarly $\max_{j\in \{1,\ldots,\ell\}}\|V'_j\|\leq 1$.
As in the proof of \expref{Lemma}{eq:round to orthogonal}, we may then argue that there exists $U_i\in \O_{m_i,n_i}$, $V_j\in \O_{p_j,q_j}$ such that 
$$f((U_i),(V_j)) \geq f((U'_i),(V'_j)) = M(U',V') = M(U,V)\,,$$
 proving the reverse inequality  $\mathrm{Opt}_\R(f) \ge \mathrm{Opt}_\R(M)$.

To conclude the proof of \expref{Proposition}{prop:usefulform} it remains to note that the algorithm of \expref{Theorem}{thm:alg}, when applied to $M$, produces in polynomial time $U,V \in \O_t$ such that $\mathrm{Opt}_\R(M)\leq O(1)\cdot M(U,V)$. Arguing as in \expref{Lemma}{eq:round to orthogonal}, $(U_i)$ and $(V_j)$ can be computed from $U,V$ in polynomial time and constitute the output of the algorithm.
\end{proof}

\subsection{Constant-factor algorithm for robust PCA problems}\label{sec:app_pca}

We start with the R1-PCA problem, as described in~\eqref{eq:R1PCA def}. Let $a_1,\ldots,a_N\in\R^n$ be given, and define $f\in\Bil(1,N;(K),(n),(1,\ldots,1),(K,\ldots,K))$ by
$$f\big(Y,(Z_1,\ldots,Z_N)\big) \eqdef \sum_{i=1}^N \sum_{k=1}^K \,(Z_i)_{1k} \, \Big(\sum_{j=1}^n\, Y_{kj}(a_i)_{j}\Big)\,.$$
The condition $Z_i\in\O_{1,K}$ is equivalent to $Z_i$ being a unit vector, while $Y\in \O_{K,n}$ is equivalent to the $K$ rows of $Y$ being orthonormal.
Using that the $\ell_2$ norm of a vector $u \in \R^K$ is equal to $\max_{v\in S^{K-1}} \langle u,v \rangle$, the quantity appearing in~\eqref{eq:R1PCA def} is equal to
\begin{equation*}
 \sup_{\substack{Y\in\O_{K,n}\\ Z_1,\ldots,Z_N \in \O_{1,K}}}  \, f\big(Y,(Z_1,\ldots,Z_N)\big)\,,
\end{equation*}
which by definition is equal to $\mathrm{Opt}_\R(f)$. An approximation algorithm for the R1-PCA problem then follows immediately from \expref{Proposition}{prop:usefulform}. The algorithm for L1-PCA follows similarly. Letting $g\in\Bil(1,NK;(K),(n),(1,\ldots,1),(1,\ldots,1))$ be defined as
$$g\big(Y,(Z_{1,1},\ldots,Z_{N,K})\big) \,=\, \sum_{i=1}^N \sum_{k=1}^K \,Z_{ik} \, \Big(\sum_{j=1}^n\, Y_{kj}(a_i)_{j}\Big)\,,$$
and using that $Z_{ik}\in\O_{1,1}$ is equivalent to $Z_{ik}\in\{-1,1\}$, the quantity appearing in~\eqref{eq:L1PCA def} is equal to
\begin{equation*}
\sup_{\substack{Y\in\O_{K,n}\\ Z_{1,1},\ldots,Z_{N,K} \in \O_{1,1}}}  \, g\big(Y,(Z_{1,1},\ldots,Z_{N,K})\big) \,=\,\mathrm{Opt}_\R(g),
\end{equation*}
which again fits into the framework of \expref{Proposition}{prop:usefulform}.

\subsection{A constant-factor algorithm for the orthogonal Procrustes problem}\label{sec:app_procrustes}

The generalized orthogonal Procustes problem was described in Section~\eqref{sec:procustes intro}. Continuing with the notation introduced there, let $A_1,\ldots,A_K$ be $d\times n$ real matrices such that the $i$-th column of $A_k$ is the vector $x^k_i\in \R^d$. Our goal is then to efficiently approximate
\begin{equation}\label{eq:defP1}
P(A_1,\ldots A_K)\eqdef \max_{U_1,\ldots U_K\in \O_d}\Big\| \sum_{k=1}^K U_k A_k \Big\|_2^2=\max_{U_1,\ldots U_K\in \O_d} \sum_{k,l=1}^K \langle U_kA_k,U_lA_l\rangle\,,
\end{equation}
where we set $\langle A,B\rangle \eqdef \sum_{i=1}^d\sum_{j=1}^n A_{ij}B_{ij}$ for every two $d\times n$ real matrices $A,B$.
We first observe that~\eqref{eq:defP1} is equal to
\begin{equation}\label{eq:defP2}
\max_{U_1,\ldots U_K\in \O_d} \max_{V_1,\ldots V_K\in \O_d} \Big\langle \sum_{k=1}^K U_k A_k,\sum_{l=1}^K V_lA_l\Big\rangle\,.
\end{equation}
It is clear that~\eqref{eq:defP2} is at least as large as~\eqref{eq:defP1}. For the other direction, using Cauchy-Schwarz,
\[
\Big\langle \sum_{k=1}^K U_kA_k,\sum_{l=1}^K V_lA_l\Big\rangle \le \Big\|\sum_{k=1}^K U_kA_k\Big\|_2\cdot \Big\|\sum_{l=1}^K V_lA_l\Big\|_2\,,
\]
so either $U_1,\ldots,U_K$ or $V_1,\ldots,V_K$ achieve a value in~\eqref{eq:defP1} that is at least as large as~\eqref{eq:defP2}.
The desired algorithm now follows by noting that~\eqref{eq:defP2} falls into the framework of \expref{Proposition}{prop:usefulform}: it suffices to define $f\in\Bil(K,K;(d,\ldots,d),(d,\ldots,d),(d,\ldots,d),(d,\ldots,d))$ by
$$f\big((U_1,\ldots,U_K),(V_1,\ldots,V_K)\big) \,=\,\Big\langle\sum_{k=1}^K U_k A_k,\sum_{l=1}^K V_lA_l\Big\rangle\,.$$
Finally, from an assignment to~\eqref{eq:defP2} we can efficiently extract an assignment to~\eqref{eq:defP1} achieving at least as
high a value by choosing the one among the $(U_k)$ or the $(V_l)$ that leads to a higher value.

\paragraph{Comparison with previous work.} To compare the approximation guarantee that we obtained with the literature, we first note that $P(A_1,\ldots,A_K)$ is attained at $U_1,\ldots,U_K\in \O_d$ if and only if $U_1,\ldots,U_K$ are maximizers of the following quantity over all $V_1,\ldots,V_K\in \O_d$:
\begin{equation}\label{eq:procrustes off diagonal}
\sum_{\substack{k,l\in \{1,\ldots,K\}\\k\neq l}} \langle V_kA_k,V_lA_l\rangle\,.
\end{equation}
  Indeed, the diagonal term $\langle V_kA_k,V_kA_k\rangle=\|V_kA_k\|_2^2$ equals $\|A_k\|_2^2$, since $V_k$ is orthogonal. Consequently, the quantity appearing in~\eqref{eq:procrustes off diagonal} differs from the quantity defining $P(A_1,\ldots,A_K)$ by an additive term that does not depend on $V_1,\ldots,V_K$. In the same vein, as already mentioned in Section~\eqref{sec:procustes intro}, $P(A_1,\ldots,A_K)$ is attained at $U_1,\ldots,U_K\in \O_d$ if and only if $U_1,\ldots,U_K$ are minimizers of the following quantity over all $V_1,\ldots,V_K\in \O_d$:
\begin{equation}\label{eq:procrustes distance version}
\sum_{k,l=1}^K \|V_kA_k-V_lA_l\|_2^2\,.
\end{equation}

While the optimization problems appearing in~\eqref{eq:defP1}, \eqref{eq:procrustes off diagonal} and~\eqref{eq:procrustes distance version} have the same exact solutions, this is no longer the case when it comes to approximation algorithms. To the best of our knowledge the 
polynomial-time 
approximability of the minimization of the quantity appearing in~\eqref{eq:procrustes distance version} was not investigated in the literature. Nemirovski~\cite{Nem07} and So~\cite{So11} studied the 
polynomial-time 
approximability of the maximization of the quantity appearing in~\eqref{eq:procrustes off diagonal}: the best known algorithm for this problem~\cite{So11} has approximation factor $O(\log(n+d+K))$. This immediately translates to the same approximation factor for computing $P(A_1,\ldots,A_K)$, which was the previously best known algorithm for this problem. Our constant-factor approximation algorithm for $P(A_1,\ldots,A_K)$ yields an approximation for the maximization of the quantity appearing in~\eqref{eq:procrustes off diagonal} that has a better approximation factor than that of~\cite{So11} unless the additive difference $\sum_{k=1}^K \|A_k\|_2^2$ is too large. Precisely, our algorithm becomes applicable in this context as long as this term is at most a factor $(1-1/C)$ smaller than $P(A_1,\ldots,A_K)$, where $C$ is the approximation constant from \expref{Proposition}{prop:usefulform}. This will be the case for typical applications in which one may think of each $A_k$ as obtained from a common $A$ by applying a small perturbation followed by an arbitrary rotation: in that case it is reasonable to expect the optimal solution to satisfy
$\langle U_kA_k , U_lA_l \rangle \approx  \| A \|_2$ for every $l,k\in \{1,\ldots,K\}$; see, \eg,~\cite[Sec.~4.3]{Nem07}.

\subsection{An algorithmic noncommutative dense regularity lemma}\label{sec:regularity}

Our goal here is to prove \expref{Theorem}{thm:regularity}, but before doing so we note that it leads to a PTAS for computing $\opt_\C(M)$ whenever $\opt_\C(M) \ge \kappa n \|M\|_2$ for some constant $\kappa>0$ (we shall use below the notation introduced in the statement of \expref{Theorem}{thm:regularity}).

The idea for this is exactly as in Section~3 of~\cite{FK99}. The main point is that given such an $M$ and $\eps\in (0,1)$ the decomposition in \expref{Theorem}{thm:regularity} only involves $T = O(1/(\kappa\eps)^2)$ terms, which can be computed in polynomial time. Given such a decomposition, we will exhaustively search over a suitably discretized set of values $a_t,b_t$ for $\Tr(A_t X)$ and $\Tr(B_t Y)$, %
respectively. For each such choice of values, verifying whether it is achievable using an $X$ and $Y$ of operator norm at most $1$ can be cast as a semidefinite program. The final approximation to $\opt_\C(M)$ is given by the maximum value of $|\sum_{t=1}^T \alpha_t a_t b_t|$, among those sequences $(a_t),(b_t)$ that were determined to be feasible.

In slightly more detail, first note that for any $X,Y$ of operator norm at most $1$ the values of $\Tr(A_t X)$ and $\Tr(B_t Y)$ lie in the complex disc with radius $n$. Given our assumption on $\opt_\C(M)$, the bound on $\alpha_t$ stated in \expref{Theorem}{thm:regularity} implies $|\alpha_t|= O(\opt_\C(M)/(\kappa n^2))$. Hence an approximation of each $\Tr(A_t X)$ and $\Tr(B_t Y)$ up to additive error $\eps\kappa n/T$ will translate to an additive approximation error to $M(X,Y)$ of $O(\eps\opt_\C(M))$. As a result, to obtain a multiplicative $(1\pm \eps)$ approximation to $\opt_\C(M)$ it will suffice to exhaustively enumerate among $(O((n\cdot T/(\eps\kappa n))^2))^{2T}$ possible values for the sequences $(a_t),(b_t)$.

Finally, to decide whether a given sequence of values can be achieved, it suffices to decide two independent feasibility problems of the following form: given $n\times n$ complex matrices $(A_t)_{t=1}^{T}$ of norm at most $1$ and $(a_t)_{t=1}^{T} \in \C^T$, does there exist $X\in M_n(\C)$ of norm at most $1$ such that $$\max_{t\in \{1,\ldots,T\}}\max\{|\mathrm{Re}(\Tr( A_t X ) - a_t)|,|\mathrm{Im}(\Tr( A_t X ) - a_t)|\}\le \frac{\eps\kappa n}{T}\, ?$$  This problem can be cast as a semidefinite program, and feasibility can be decided in time that is polynomial in $n$ and $T$.

\begin{proof}[Proof of \expref{Theorem}{thm:regularity}]
The argument is iterative. Assume that $\{A_t,B_t\}_{t=1}^{\tau-1}\subseteq \U_n$ have already been constructed. Write $M_1\eqdef M$ and $$M_{\tau} \eqdef M - \sum_{t=1}^{\tau-1} \alpha_t (A_t\otimes B_t)\,.$$ If $\opt_\C(M_\tau) \leq \eps \opt_\C(M)$ then we may stop the construction. Otherwise, by \expref{Theorem}{thm:alg}
and multiplying by an appropriate complex phase if necessary, we can find $A_\tau,B_\tau\in \U_n$ such that
\begin{equation}\label{eq:tau step}
 M_\tau(A_\tau,B_\tau) \,\geq\, \Big(\frac{1}{2} - \frac{\eps}{2}\Big) \opt_\C(M_\tau)\,,
 \end{equation}
with the left-hand side of~\eqref{eq:tau step} real. Set 
$$\alpha_\tau \eqdef \frac{M_\tau(A_\tau,B_\tau)}{\|A_\tau\|_2^2\cdot \|B_\tau\|_2^2}\,,$$
 and define $M_{\tau+1} \eqdef M_{\tau} - \alpha_\tau(A_\tau\otimes B_\tau)$. 
By expanding the square we have
\begin{align}\label{eq:frobenius}
\nonumber\big\| M_{\tau+1} \big\|_2^2 &= \big\| M_{\tau} \big\|_2^2  - \frac{M_\tau (A_\tau , B_\tau)^2}{ \|A_\tau\|_2^2 \|B_\tau \|_2^2}\\ &\leq \big\| M_{\tau} \big\|_2^2 - \frac{(1-\eps)^2}{4n^2}\opt_\C(M_\tau)^2\nonumber \\
&\leq \big\| M_{\tau} \big\|_2^2 - \frac{\eps^2(1-\eps)^2}{4n^2}\opt_\C(M)^2.
\end{align}
It follows from~\eqref{eq:frobenius} that as long as this process continues, $\|M_\tau\|_2^2$ decreases by an additive term of at least $\eps^2(1-\eps)^2 \opt_\C(M)^2/(4n^2)$ at each step. This process must therefore terminate after at most $T\leq 4n^2 \|M\|_2^2/( \eps (1-\eps)\opt_\C(M))^2$ steps. It also follows that $\|M_\tau\|_2\leq\|M\|_2$ for every $\tau$, so that
\[
|\alpha_\tau| \leq \frac{\|M_\tau\|_2\|A_\tau\|_2\|B_\tau\|_2}{\|A_\tau\|_2^2\cdot \|B_\tau\|_2^2} = \frac{\|M_\tau \|_2}{n}\le \frac{\|M\|_2}{n}\,,
\]
where the first inequality follows from Cauchy-Schwarz, and the equality uses that $A_\tau,B_\tau\in\U_n$.

The ``moreover'' part of \expref{Theorem}{thm:regularity} follows immediately by using the part of \expref{Theorem}{thm:alg} pertaining to $\opt_\R(M)$  (note that the specific approximation factor does not matter here).
\end{proof}

\bigskip

\noindent{\bf Acknowledgements.} We thank  Daniel Dadush and Raghu Meka for useful discussions.

\bibliographystyle{tocplain}   %
\bibliography{v010a011}

\begin{tocauthors}
\begin{tocinfo}[naor]
 Assaf Naor\\
Professor\\
Courant Institute of Mathematical Sciences, New York University, New York, NY\\
 naor\tocat{}cims\tocdot{}nyu\tocdot{}edu \\
 \url{http://www.cims.nyu.edu/~naor/}
\end{tocinfo}
\begin{tocinfo}[regev]
  Oded Regev\\
Professor\\
Courant Institute of Mathematical Sciences, New York University, New York, NY\\
 regev\tocat{}cims\tocdot{}nyu\tocdot{}edu \\
\url{http://www.cims.nyu.edu/~regev/}
\end{tocinfo}
\begin{tocinfo}[vidick]
  Thomas Vidick\\
Assistant Professor\\
California Institute of Technology, Pasadena, CA\\
 vidick\tocat{}cms\tocdot{}caltech\tocdot{}edu \\
\url{http://cms.caltech.edu/~vidick}
\end{tocinfo}
\end{tocauthors}

\begin{tocaboutauthors}
\begin{tocabout}[naor]
\textsc{Assaf Naor}'s research focuses on analysis and metric geometry, and their
interactions with approximation algorithms and complexity theory. He
received his \phd\ from the \href{http://new.huji.ac.il/}{Hebrew University} in 2002, advised by Joram
Lindenstrauss. He is a Professor of Mathematics and Computer Science at the
\href{http://www.cims.nyu.edu/}{Courant Institute of Mathematical Sciences} of New York University, where he
has been a faculty member since 2006. Prior to joining the Courant Institute
he was a researcher at the Theory Group of Microsoft Research in Redmond WA\@.
Starting fall 2014 he will be a Professor of Mathematics at \href{http://www.princeton.edu}{Princeton
University}.
\end{tocabout}
\begin{tocabout}[regev]
\textsc{Oded Regev} graduated from
  \href{http://www.tau.ac.il/}{Tel Aviv University} in 2001 under the
  supervision of \href{http://www.cs.tau.ac.il/~azar/}{Yossi Azar}.
  He spent two years as a postdoc at the \href{http://www.ias.edu/}{Institute for
    Advanced Study}, Princeton, and one year at the
  \href{http://www.berkeley.edu/}{University of California, Berkeley}.
  He recently joined the Courant Institute of Mathematical Sciences and is still trying to
	get used to life in NYC.  %
  His research interests include computational
  aspects of lattices, quantum computation, and other topics in theoretical computer
  science. 
\end{tocabout}
\begin{tocabout}[vidick]
\textsc{Thomas Vidick} graduated from \href{http://www.berkeley.edu}{UC Berkeley} in 2011; his advisor was \href{http://www.cs.berkeley.edu/~vazirani/}{Umesh Vazirani}. His thesis focused on the study of quantum entanglement in multi-prover interactive proof systems and in quantum cryptography. After a postdoctoral scholarship at \href{http://www.mit.edu}{MIT} under the supervision of \href{http://www.scottaaronson.com/}{Scott Aaronson}, he moved back to sunny California. He is currently an assistant professor in \href{http://www.caltech.edu/}{Caltech}'s department of \href{http://www.cms.caltech.edu/}{Computing and Mathematical Sciences}, where his research is stimulated by the humbling mark left by
the previous occupants of his his office and his neighbors'---\href{http://en.wikipedia.org/wiki/Alexei_Kitaev}{Alexei Kitaev} and \href{http://en.wikipedia.org/wiki/Richard_Feynman}{Richard Feynman}.
\end{tocabout}
\end{tocaboutauthors}

\end{document}